%% file: main-arxiv.tex
\documentclass[11pt,a4paper]{article}
\usepackage{a4wide}
\usepackage[utf8]{inputenc}
\usepackage[T1]{fontenc}
\usepackage{lmodern}
\usepackage{microtype}
\usepackage{amsmath,amssymb,amsthm}
\usepackage[title, titletoc]{appendix}
\usepackage{thmtools}
\usepackage{graphicx}
\usepackage{tikz}
\usetikzlibrary{arrows.meta}
\usepackage{subcaption}
\usepackage{xcolor}
\usepackage{todonotes}
\usepackage{amsopn}
\usepackage{enumitem}

\definecolor{linkblue}{RGB}{0,70,140}
\definecolor{citegreen}{RGB}{0,100,60}
\definecolor{urlblue}{RGB}{0,90,160}

\usepackage[colorlinks=true,linkcolor=linkblue,citecolor=citegreen,urlcolor=urlblue]{hyperref}
\usepackage[nameinlink,capitalise,noabbrev]{cleveref}

\crefname{enumi}{item}{items}
\Crefname{enumi}{Item}{Items}

\title{Constant-factor approximation of $\MinCostHom$\\ with a conservative majority polymorphism\thanks{This work was supported by the European Research Council (ERC) under the European Union’s Horizon 2020 research and innovation programme (grant agreement No 714532), by UKRI EP/X024431/1, by the Polish National Agency for Academic Exchange, and by the National Science Centre, Poland under the  Weave-UNISONO  call  in  the  Weave  programme 2021/03/Y/ST6/00171. For the purpose of open access, the authors have applied a CC BY public copyright licence to any Author Accepted Manuscript version arising from this submission. All data is provided in full in the results section of this paper. The second author thanks (and vice versa!) the first for visiting him in 2022, when this collaboration started.}}
\author{
Marcin Kozik\\
Jagiellonian University\\
\texttt{marcin.kozik@uj.edu.pl}
\and
Stanislav \v{Z}ivn\'y\\
University of Oxford\\
\texttt{standa.zivny@cs.ox.ac.uk}
}

\date{} 

\declaretheorem[name=Theorem,numberwithin=section]{theorem}
\declaretheorem[name=Observation,numberwithin=section]{observation}

\newtheorem{proposition}[theorem]{Proposition}

\declaretheoremstyle[
    headfont=\bfseries,
    bodyfont=\itshape,
    notefont=\bfseries,
    notebraces={``}{''},
    headpunct={.},
    postheadspace=0.5em,
    headformat=\NAME\NOTE
]{namedlemmastyle}
\declaretheorem[style=namedlemmastyle,numbered=no,name=Lemma]{namedlemma}
\newcommand{\namedlemcref}[1]{\hyperref[#1]{Lemma ``\nameref*{#1}''}}

\theoremstyle{definition}
\newtheorem{definition}[theorem]{Definition}

\theoremstyle{remark}

\newcommand{\CSP}{\ifmmode\text{\normalfont\textsc{CSP}}\else\textsc{CSP}\fi}
\newcommand{\MinCSP}{\ifmmode\text{\normalfont\textsc{MinCSP}}\else\textsc{MinCSP}\fi}
\newcommand{\cUGC}{\ifmmode\text{\normalfont\textsc{cUG}}\else\textsc{cUG}\fi}
\newcommand{\MinCostHom}{\ifmmode\text{\normalfont\textsc{MinCostCSP}}\else\textsc{MinCostCSP}\fi}
\newcommand{\Conn}[1]{\operatorname{Conn}(#1)}
\newcommand{\Cen}[1]{\operatorname{Cen}(#1)}
\newcommand{\relstr}[1]{\ifmmode\text{\normalfont\textbf{#1}}\else\textbf{#1}\fi}
\newcommand{\constr}[1]{{\mathcal #1}}
\newcommand{\inst}[1]{{\mathcal #1}}
\providecommand{\OPT}{\operatorname{OPT}}
\providecommand{\SDP}{\operatorname{SDP}}
\DeclareMathOperator{\proj}{proj}
\DeclareMathOperator{\ar}{ar}
\DeclareMathOperator{\cost}{cost}
\DeclareMathOperator{\pol}{Pol}

\newcommand{\blambda}{\ensuremath{\mbox{\boldmath $\lambda$}}}
\newcommand{\bkappa}{\ensuremath{\mbox{\boldmath $\kappa$}}}

\begin{document}

\maketitle

\begin{abstract}
\input{abstract}

\end{abstract}

\input{paper}

\begin{appendices}
    \input{ACproof}
    \input{tractability-rounding-and-cost}
    \input{tractability-solvability}
    \input{cycles-are-hard}

\end{appendices}

\section*{AI Methodology}

ChatGPT 5.5 was used to generate LaTeX source of some of the proofs, and to
suggest possible issues with arguments.
As mentioned before, ChatGPT also gave the construction presented in~\Cref{sec:directed-cycle-gap}.

{\small
\bibliographystyle{plainurl}
\bibliography{bib}
}

\end{document}

%% file: abstract.tex
  For a relational structure $\relstr{A}$, the \emph{Minimum Cost Constraint
  Satisfaction Problem} is the following problem denoted by
  $\MinCostHom(\relstr{A})$: Given an instance of $\CSP(\relstr{A})$ with
  rational costs on variable-value pairs, find a solution to the instance
  minimizing the sum of the chosen costs. For the exact minimization, a
  classification of $\MinCostHom(\relstr{A})$ in terms of $\relstr{A}$ was
  established by Takhanov~[STACS'10].
    
  We focus on constant-factor approximations of $\MinCostHom(\relstr{A})$.
  DeHaan, Huang, and Lee recently showed that if $\relstr{A}$ fails to admit a
  conservative near-unanimity polymorphism then $\MinCostHom(\relstr{A})$ is not
  constant-factor approximable~[APPROX'25]. We provide a significant step
  towards a full classification, by proving a dichotomy for structures
  $\relstr{A}$ admitting a conservative majority (also known as
  $3$-near-unanimity) polymorphism. Our dichotomy criterion is not in terms of
  an algebraic condition on $\relstr{A}$ but we show that this is unavoidable.
  We include a simple argument proving that no such condition exists.

%% file: paper.tex
\section{Introduction}\label{sec:intro}

Satisfiability of formulas with three literals per clause (3-SAT), Maximum Cut (MaxCut), and Minimum Vertex Cover (MinVC) are all examples
of Boolean \emph{Constraint Satisfaction Problems}
(CSPs)~\cite{Creignouetal:siam01}. Given a set of variables and a set of
constraints, one seeks an assignment of $0$s and $1$s to the variables with the
goal of satisfying all the constraints (3-SAT), maximizing the
number of satisfied constraints (MaxCut), and minimizing the
number of variables assigned the value $1$ while simultaneously satisfying all the
constraints (MinVC).\footnote{We call a CSP \emph{Boolean} if the domain of each variable is of size two.} 

This paper is concerned with CSPs parameterized by the set of allowed
constraint relations, known as non-uniform CSPs~\cite{Kolaitis00:jcss}, 
fixed-template CSPs~\cite{Feder98:sicomp}, or CSPs with a fixed constraint
language~\cite{Jeavons97:jacm}. 
Formally, one fixes a (relational) structure $\relstr{A}=(A;R_1,\ldots,R_p)$,
where $A$ is a finite domain of values for the variables and each $R_i$ is a
relation on $A$. We denote by $\CSP(\relstr{A})$ the set of CSP instances whose
constraints only use relations from $\relstr{A}$.
For example, 3-SAT is captured by a CSP with eight (in fact four
suffice) relations corresponding to the eight types of 3-clauses, MaxCut is
captured by a CSP with one relation, namely the binary disequality relation
$\{(0,1),(1,0)\}$, and MinVC is captured by a CSP with the 
relation $\{(0,1),(1,0),(1,1)\}$.

The three examples deal with different computational problems. 3-SAT is an example
of a \emph{decision} Boolean CSP. All decision Boolean CSPs were classified as solvable
in polynomial time or NP-complete by Schaefer back in
1978~\cite{Schaefer78:complexity}. MaxCut is an example of an \emph{optimization}  Boolean
CSP. All optimization Boolean CSPs were classified with respect to exact
solvability by Creignou~\cite{Creignou95:jcss} and with respect to
approximability,  both for minimization and maximization, by Khanna, Sudan,
Trevisan, and Williamson~\cite{Khanna01:sicomp} (see also the monograph by
Creignou, Khanna, and Sudan~\cite{Creignouetal:siam01}). Finally, MinVC 
is one of the simplest examples of a Boolean CSP that mixes decision and
optimization. In other words, it is an optimization problem with strict
constraints. Such CSPs are called MinOnes
in~\cite{Khanna01:sicomp,Creignouetal:siam01} and TMIN by Khanna and Motwani
in~\cite{Khanna96:stoc}.
All Boolean MinOnes, as well as
MaxOnes, were classified with respect to approximability
in~\cite{Khanna01:sicomp}.

While by now we have a very good understanding of Boolean CSPs, the situation is
significantly more complicated for CSPs over larger, non-Boolean domains, also
known as alphabets. Over the years, various fragments of non-Boolean CSPs have
been studied. Two prominent examples that were important in the development of
the field are \emph{graph} CSPs and \emph{conservative} CSPs. The former are
CSPs of the form $\CSP(\relstr{H})$, where $\relstr{H}$ is a graph, i.e., a
single binary symmetric relation. The latter are CSPs that allow arbitrary
restrictions of the domains of individual variables. We now survey what is known
for non-Boolean CSPs, mentioning important special cases.

\smallskip

For \emph{decision} CSPs, Hell and Ne\v{s}et\v{r}il established a 
classification of graph CSPs under the name of
$\relstr{H}$-coloring~\cite{HellN90}.
A complexity
classification for conservative CSPs was established by
Bulatov~\cite{Bulatov11:tocl} (with simple proofs later by
Barto~\cite{Barto11:lics} and Bulatov~\cite{Bulatov16:jcss}), following earlier
classifications for graphs by Feder, Hell, and
Huang~\cite{Feder03:jgt} and for digraphs by Hell and Rafiey~\cite{Hell11:soda}.
After a lot of partial results and important developments, a dichotomy for \emph{all} finite-domain CSPs was established by Bulatov~\cite{Bulatov17:focs} and
Zhuk~\cite{Zhuk20:jacm} via the so-called algebraic approach~\cite{BKW17}, thus
confirming in the affirmative the Feder--Vardi
conjecture~\cite{Feder98:sicomp} and its algebraic formulation by Bulatov,
Jeavons, and Krokhin~\cite{BulatovJeavonsKrokhin2005}.

For \emph{optimization} CSPs, Thapper and \v{Z}ivn\'y classified all
finite-domain CSPs with respect to exact solvability~\cite{tz16:jacm}, building
on earlier classifications for exact solvability of CSPs on Boolean
domains~\cite{Creignou95:jcss}, CSPs on three-element
domains~\cite{JonssonKK06:3valued,hkp14:sicomp}, and conservative
CSPs~\cite{DeinekoJKK08,kz13:jacm}.

For approximation, there is a huge difference between minimization and maximization, as
we shall elaborate on next.
All \emph{Maximum} CSPs (MaxCSPs) admit a constant-factor approximation
algorithm. A celebrated result of Raghavendra~\cite{Raghavendra08:everycsp}
established that the basic SDP relaxation achieves the best approximation ratio
for all MaxCSPs under Khot's Unique Games Conjecture (UGC)~\cite{Khot02stoc}, building
on earlier works of Khot, Kindler, Mossel, and O'Donnell~\cite{Khot07:sicomp}
and Austrin~\cite{Austrin10:sicomp}. While one can in principle round the SDP
solution~\cite{Raghavendra10:stoc}, the precise approximation factors are open
for most CSPs, cf. the recent work by Brakensiek, Huang, Potechin, and
Zwick~\cite{BHPZ25:sidma,Brakensiek23:focs}.

Not all \emph{Minimum} CSPs (MinCSPs) admit a constant-factor approximation. In
fact, MinCSPs attain various levels of
approximability~\cite{Agarwal05:stoc,Hastad08:cc,Guruswami12:toc,Guruswami16:toc},
already on Boolean domains~\cite{Creignouetal:siam01}. Constant-factor
approximable MinCSPs were studied, e.g., by Ene, Vondr\'ak, and
Wu~\cite{Ene13:soda} and by Dalmau, Krokhin, and Manokaran~\cite{Dalmau18:jcss},
in relation to approximability via the basic linear programming relaxation. We
note that constant-factor approximability of MinCSPs is closely related to
robustly solvable CSPs~\cite{Guruswami12:toc,Barto16:sicomp} (a notion
introduced by Zwick~\cite{Zwick98:stoc}) with linear
loss~\cite{Dalmau13:robust}.

\begin{center}
  \emph{What about approximation of CSPs that mix optimization with strict constraints?}
\end{center}

The simplest {optimization CSP with strict constraints} is the 
\emph{Minimum Cost Constraint Satisfaction Problem} ($\MinCostHom$),
introduced (for graphs) 
by Gutin, Rafiey, Yeo,
and Tso in~\cite{Gutin06:dam-repair}. 
Fix a structure $\relstr{A}$.  
Then, $\MinCostHom(\relstr{A})$ is the following problem:
Given a $\CSP(\relstr{A})$ instance with variables $V$ over domain $A$ and 
a cost function $c_x:A\to\mathbb{Q}_{\geq 0}\cup\{\infty\}$
for every variable $x\in V$,
the goal is to find an assignment $\alpha:V\to A$ that
satisfies all the constraints and minimizes $\sum_{x\in V}c_x(\alpha(x))$.
As an example, MinVC is a Boolean $\MinCostHom(\relstr{A})$ with a single relation, where
$\relstr{A}=(\{0,1\};\{(0,1),(1,0),(1,1)\})$ and
the costs are $c_x(0)=0$ and $c_x(1)=1$ as we interpret the value $1$ as selecting
the vertex represented by variable $x$
in the cover.
MinVC is an example of a special type of $\MinCostHom$s with \emph{injective} cost;
i.e., $c_x(a)\neq c_x(b)$ for every $x$ and $a\neq b$. 
Such $\MinCostHom$s have been studied under the names of StrictCSPs by Kumar, Manokaran,
Tulsiani, and Vishnoi~\cite{Kumar11:soda} and MinSolution by Jonsson, Kuivinen,
and Nordh~\cite{Jonsson08:sicomp}. 

Since cost functions in $\MinCostHom$s can take on infinite costs,
$\MinCostHom$s are a generalization of conservative CSPs. 
Thus, $\MinCostHom(\relstr{A})$ can be constant-factor
approximable in polynomial time only if $\CSP(\relstr{A}')$ is solvable in
polynomial time, where
$\relstr{A}'$ is $\relstr{A}$ expanded with all unary relations on $A$. Moreover,
without loss of generality we can consider $\relstr{A}'$ instead of $\relstr{A}$
and thus can assume that all \emph{polymorphisms} of
$\relstr{A}$ are conservative. Intuitively, a polymorphism of $\relstr{A}$ is a
closure operation on the solution space of any instance of
$\MinCostHom(\relstr{A})$ (for a precise definition
cf.~\Cref{sec:preliminaries}). We call a polymorphism $f:A^k\to A$ conservative
if $f(x_1,\ldots,x_k)\in\{x_1,\ldots,x_k\}$.

\medskip
The fundamental question of which $\MinCostHom$s admit 
constant-factor approximation is still open despite progress 
on exact
solvability~\cite{Gutin08:ejc-mincostdichotomy,Hell12:sidma,Takhanov10:stacs}
(discussed in detail in related work later) and progress on the following
special cases:

\begin{itemize}[leftmargin=1.3em]
  \item 
    For Boolean structures $\relstr{A}$,
    a constant-factor approximation dichotomy of $\MinCostHom(\relstr{A})$ 
    follows from the already mentioned work of
    Khanna et al.~\cite{Khanna01:sicomp}. The Boolean domain is significantly
    easier to understand compared to larger domains, as there is a complete description
    of relational clones on Boolean domains by Post~\cite{Post41}.

  \item 
    For bipartite graphs $\relstr{A}$,
    a constant-factor approximation dichotomy of $\MinCostHom(\relstr{A})$ 
    was established by Hell, Mastrolilli, Nevisi, and
    Rafiey~\cite{Hell12:esa-approximation}. In this case, the bipartiteness of
    $\relstr{A}$ imposes enough structure to limit the approximable cases and
    allows for simple gadgets to show
    inapproximability~\cite{Mastrolilli13:dam}.

  \item Recent work of DeHaan, Huang, and Lee showed several interesting
  results~\cite{DeHaan25:approx}. 

    \begin{itemize}[leftmargin=1.3em]
      \item Building on the already
        mentioned~\cite{Dalmau18:jcss}, the authors of~\cite{DeHaan25:approx} showed that $\MinCostHom(\relstr{A})$ is not
    constant-factor approximable unless $\relstr{A}$ admits a conservative near-unanimity
    polymorphism. A $k$-ary operation $f:A^k\to A$ is called a
    \emph{near-unanimity} operation if it satisfies, for every $a,b \in A$,
    $f(a,\ldots,a,b)=f(a,\ldots,a,b,a)=\cdots=f(b,a,\ldots,a)=a$.
        This \emph{necessary} condition for constant-factor approximability 
        of $\MinCostHom(\relstr{A})$ 
        is known to be sufficient for Boolean
        $\relstr{A}$~\cite{Khanna01:sicomp}. Consequently, the existence of a
        conservative near-unanimity polymorphism of $\relstr{A}$ governs the constant-factor approximability of $\MinCostHom(\relstr{A})$ for Boolean $\relstr{A}$.
        It would be tempting to conjecture that the same condition governs the
        constant-factor approximability of $\MinCostHom(\relstr{A})$ for
        non-Boolean $\relstr{A}$ as well.

  \item A near-unanimity operation $m:A^3\to A$ of three
    arguments is known as a \emph{majority} operation as the identity amounts
        to, for every $a,b \in A$, $m(a,a,b)=m(a,b,a)=m(b,a,a)=a$. 

    The authors of~\cite{DeHaan25:approx} showed that
        $\MinCostHom(\relstr{A})$ is constant-factor approximable if
        $\relstr{A}$ admits a very special type of
    conservative majority polymorphism, known as the dual discriminator. In this
        case, $\relstr{A}$ has an explicit description of possible (and rather
        restricted) relations~\cite{Cooper94:characterising}.
        However, \cite{DeHaan25:approx} 
        observed that for a simple $\relstr{A}$, a triangle with a single loop
        (cf.~\Cref{fig:counter}), $\MinCostHom(\relstr{A})$ is not constant-factor approximable
        although $\relstr{A}$ admits a conservative majority polymorphism. 
        Consequently, the existence of 
        a conservative near-unanimity polymorphism (or even a conservative majority
        polymorphism) is \emph{not sufficient} for constant-factor
        approximability of $\MinCostHom(\relstr{A})$ for non-Boolean
        $\relstr{A}$ (already on a three-element domain).

  \item 
    For structures $\relstr{A}$ including all permutation relations on $A$,
    \cite{DeHaan25:approx} established a constant-factor approximation dichotomy
    of $\MinCostHom(\relstr{A})$; the approximation boundary is delineated by the
        existence of a conservative majority polymorphism of $\relstr{A}$ in this case.

  \end{itemize}

  \item 
    For graphs $\relstr{A}$, \cite{Rafiey19:icalp} claimed a constant-factor
    approximation dichotomy of $\MinCostHom(\relstr{A})$, with the approximation
    boundary being delineated by the presence of a conservative majority polymorphism of
    $\relstr{A}$.
    However, this result is incorrect, as observed in~\cite{DeHaan25:approx}:
    $\MinCostHom(\relstr{G})$ for the structure $\relstr{G}$
    from~\Cref{fig:counter} gives a
    counterexample (assuming P$\neq$NP and the UGC). 
\end{itemize}

For which structures $\relstr{A}$ is $\MinCostHom(\relstr{A})$ constant-factor
approximable? Can the approximability boundary of 
$\MinCostHom(\relstr{A})$ be expressed in terms of
algebraic properties of $\relstr{A}$, just like for (many variants of) CSPs and for the studied
special cases of $\MinCostHom$s discussed above?

\paragraph{Contributions}
We make a significant step towards resolving which $\MinCostHom$s admit
constant-factor approximation. In particular, we resolve the approximability of
$\MinCostHom$s admitting a conservative majority polymorphism. 
Majority-closed $\MinCostHom$s have been considered in several manuscripts,
including~\cite{Rafiey19:icalp}, although only for graphs. All contain errors
and incorrect classifications, showing the problem’s difficulty.

Unlike relations closed under the dual discriminator~\cite{DeHaan25:approx}, relations closed under a conservative majority polymorphism have no known generic description. We present the first algebraic algorithm for $\MinCostHom$s, i.e., one that works without such a description. Results for other CSP variants show combinatorial methods, including those employed in~\cite{DeHaan25:approx}, suffice only in very special cases.
Our paper is the first on $\MinCostHom$s to identify an ``obstacle'' that must be excluded, in addition to the algebraic condition, to guarantee tractability, and to prove that it implies hardness.

In CSP research, going from the dual discriminator to all conservative majorities is often a big step because no explicit description of the relations exists. E.g., for the CSP-based ideal membership problem, conservative majorities remain open; only the dual discriminator case is settled~\cite{Bulatov22:stoc}. 

\medskip
We state our main result as a single theorem: a dichotomy. First, we
give a definition.
For $k\geq 2$, let $Z_k=\{(a,a+1 \bmod k):a\in\{0,\dots,k-1\}\}$. The
\emph{Cyclic Unique Games} problem has alphabet $\{0,\dots,k-1\}$ and directed
constraints $Z_k$. The deletion version, denoted by $\cUGC(k)$,
charges cost $0$ to tuples in $Z_k$ and positive cost to tuples outside
of $Z_k$. Put differently, $\cUGC(k)$ is a MinCSP in which every variable has
domain $\{0,\ldots,k-1\}$ and the only allowed constraint gives zero cost to tuples in $Z_k$ and
positive cost to tuples not in $Z_k$.\footnote{For constant-factor
approximation, allowing positive costs and requiring all costs to equal 1 are
equivalent.}

\begin{theorem}\label{thm:DichotomyIntro}
Let $\relstr{A}$ be a finite relational structure admitting a conservative majority polymorphism. Then one of the following holds:
\begin{itemize}
    \item $\MinCostHom(\relstr{A})$ admits a randomized polynomial-time constant-factor approximation algorithm;
    \item a fixed-alphabet cyclic Unique Games deletion problem reduces to $\MinCostHom(\relstr{A})$.
\end{itemize}
\end{theorem}

In order to prove~\Cref{thm:DichotomyIntro}, we analyze instances of $\MinCostHom(\relstr{A})$ via techniques from the algebraic
approach to CSPs~\cite{BKW17}, including the notion of absorption~\cite{Barto12:lmcs}. 
The crux of the proof, both
in the correctness analysis of our algorithm and in the reduction from $\cUGC(k)$,
lies in the intricate interplay between two notions of consistency, arc-consistency and cycle-consistency. 

The algorithm in~\Cref{thm:DichotomyIntro}
uses an LP rounding based on a sequence of
exponentially decreasing thresholds; it is significantly more involved
than the algorithms (one greedy, one LP-based) for $\MinCostHom$s with a dual discriminator polymorphism
considered in~\cite{DeHaan25:approx}. 

The hardness result is obtained by a reduction from $\cUGC(k)$. We note that this is
different from how inapproximability is established in~\cite{DeHaan25:approx},
where gadget reductions from the Nearest Codeword
Problem~\cite{dumer2003hardness,cheng2012deterministic} (for $\MinCostHom$s of
unbounded width) and from MinVC in hypergraphs~\cite{DinurGKR05} (for
$\MinCostHom$s of bounded width but with no near-unanimity polymorphism) are
established.

Why is $\cUGC(k)$ NP-hard to approximate within any constant factor?
For $k=2$, $\cUGC(k)$ is the MinUnCut problem~\cite{Agarwal05:stoc}, asking to partition the
vertex set of a given graph while minimizing the number of edges with both
endpoints in the same part.
It is NP-hard to approximate (edge-weighted) MinUnCut
within any constant factor assuming the UGC~\cite{Khot07:sicomp}.
For $k\geq 2$, we believe that NP-hardness of approximation within any constant factor
is known. In fact, ChatGPT suggested that this is folklore but
we have not found any reference. For completeness, we provide a proof, which
uses Raghavendra's result~\cite{Raghavendra08:everycsp}. It boils down to 
constructing solutions to the basic SDP relaxation with large integrality gaps.
We present a construction that was suggested by ChatGPT and that uses directed
cycles, motivated by the work of Feige and Schechtman on integrality gaps for MaxCut~\cite{Feige02:rsa}.

\medskip
The general problem of classifying all $\MinCostHom$s remains open, but we
settled the difficult case of conservative majorities by introducing new ideas,
which we believe are in the right direction.
Finally, the message ``no algebraic classification of $\MinCostHom$s is
possible'' suggests that closeness to a full classification should not be
measured by majority versus NU polymorphisms.

\paragraph{Related work}

Exact solvability of $\MinCostHom(\relstr{A})$ was established for graphs $\relstr{A}$
by Gutin, Hell, Rafiey, and Yeo~\cite{Gutin08:ejc-mincostdichotomy},
for digraphs $\relstr{A}$ by Hell and Rafiey~\cite{Hell12:sidma}, and finally
for all relational structures $\relstr{A}$ by Takhanov~\cite{Takhanov10:stacs}.
As a matter of fact, the complexity of finding an optimal solution follows from
a more general framework, called conservative general-valued CSPs, whose
complexity was established by Kolmogorov and \v{Z}ivn\'y~\cite{kz13:jacm},
see also~\cite{tz17:sicomp} for a simpler proof.
Going beyond $\MinCostHom$s, the works of Kozik and Ochremiak~\cite{Kozik15:icalp} and Kolmogorov, Krokhin, and
Rol\'inek~\cite{Kolmogorov17:sicomp} classified all general-valued CSPs with
respect to exact solvability.

OpenAI recently announced a proof of the Unique Games
Conjecture~\cite{OpenAI26:UGC}. If correct then our main result gives a P
versus NP-hard dichotomy.

\paragraph{Structure of the paper}

The paper is structured as follows.
The following section contains preliminaries and definitions.
In \Cref{sec:result} we reformulate the main result,
and define a property that distinguishes the two cases of the dichotomy.
In the same section, we exhibit~%
an example showing that the distinguishing feature cannot be a term-condition
(the observation is due to Zarathustra Brady).
\Cref{sec:hardness} contains the proof of hardness,
and \Cref{sec:tractability} contains the proof of tractability.
The most technical parts of the proofs are deferred to the appendices:
\Cref{sec:AC} proves the core result for the hardness case
of~\Cref{thm:DichotomyIntro},
\Cref{sec:tractability-rounding-cost} and \Cref{sec:tractability-solvability}
give details of the tractability case of~\Cref{thm:DichotomyIntro}, and finally
\Cref{sec:directed-cycle-gap} establishes inapproximability of $\cUGC(k)$ for
$k\geq 2$.

\section{Preliminaries}\label{sec:preliminaries}

Fix a finite set $A$. We call $R\subseteq A^m$ a \emph{relation} on $A$ of
arity $m=\ar(R)$, and say that $R$ is \emph{full} if $R=A^m$. 
If $m=2$ then $R$ is called \emph{binary}. A \emph{relational
structure} $\relstr{A}=(A;R_1,\ldots,R_p)$ is a collection of relations 
$R_i$, $1\leq i\leq p$, 
on the same domain $A$. A relational structure $\relstr{A}$ is called
\emph{finite} if its domain $A$ is finite. All relational
structures in this paper are finite.

A $k$-ary operation $f:A^k\to A$ is called a \emph{polymorphism} of
$\relstr{A}=(A;R_1,\ldots,R_p)$ if $f$ preserves all relations in $\relstr{A}$; 
i.e., for every $1\leq i\leq p$ and every possible $m\times k$
matrix whose columns belong to $R_i$, where $m=\ar(R_i)$, the application of $f$ to the rows of
the matrix gives a tuple that belongs to $R_i$. We denote by $\pol(\relstr{A})$
the set of all polymorphisms of $\relstr{A}$.

\begin{definition}
  Fix a relational structure $\relstr{A}=(A;R_1,\ldots,R_p)$. An instance of the \emph{Minimum
Cost Constraint Satisfaction Problem} over $\relstr{A}$, denoted by $\MinCostHom(\relstr{A})$, is a
  triple $(V,\mathcal{C},\{c_x\}_{x\in V})$, where $V$ is a 
set of variables; $\mathcal{C}$ is a set of constraints, and each constraint
$(\mathbf{x},R)\in\mathcal{C}$ is a pair with $\mathbf{x}\in V^m$ an $m$-tuple
of variables from $V$, and $R=R_i$ a relation of arity $\ar(R)=m$ for some
$1\leq i\leq p$; and $c_x:A\to\mathbb{Q}_{\geq 0}\cup\{\infty\}$. An
assignment
$\alpha:V\to A$ is a \emph{solution} if $\alpha(\mathbf{x})\in R$
for every $(\mathbf{x},R)\in\mathcal{C}$, where $\alpha$ is applied
componentwise. The \emph{cost} of a solution $\alpha$
is defined as $\cost(\alpha)=\sum_{x\in V}c_x(\alpha(x))$. The goal is to
find a solution of minimum cost.
\end{definition}

If $c_x=0$ for every $x\in V$ then $\MinCostHom(\relstr{A})$ is just the \emph{Constraint Satisfaction Problem} 
$\CSP(\relstr{A})$~\cite{Feder98:sicomp}.
It is known that the computational complexity of $\CSP(\relstr{A})$ is
determined by
$\pol(\relstr{A})$~\cite{Jeavons97:jacm,BulatovJeavonsKrokhin2005} in the sense
that if there is a clone homomorphism from $\pol(\relstr{A})$ to
$\pol(\relstr{B})$ then $\CSP(\relstr{B})$ polynomial-time reduces to $\CSP(\relstr{A})$ (via
so-called ``gadget reductions'', cf.~\cite{BKW17}).

If $c_x(a)\in\{0,\infty\}$ for every $x$ and $a$ then $\MinCostHom(\relstr{A})$ is $\CSP(\relstr{A}')$, where $\relstr{A}'$ is $\relstr{A}$ expanded with all unary relations on $A$; as we discussed in~\Cref{sec:intro}, such CSPs are known as \emph{conservative} CSPs~\cite{Bulatov11:tocl}. It is easy to see that the costs 0 and $\infty$ in the cost functions $c_x$ allow for simulating all unary relations. Thus, without loss of generality we can assume that every polymorphism of $\relstr{A}$ is conservative~\cite{DeHaan25:approx}, and we can also assume that the range of $c_x$ is included in $\mathbb{Q}_{\geq 0}$ (as domain values $a$ with $c_x(a)=\infty$ cannot be part of any solution).

\section{The result}\label{sec:result}
We now reformulate the result in purely combinatorial terms;
in order to do so, we need to introduce standard definitions and recall some known facts.

Let us begin with an arbitrary finite relational structure $\relstr{A}$ with a conservative majority polymorphism.
We first add all unary relations and then all binary relations pp-definable over the resulting structure.
These expansions preserve constant-factor approximability, by the preceding discussion and~\cite[Lemma~30]{DeHaan25:approx}.
By Baker--Pixley~\cite{BakerPixley}, 
we can replace every relation in $\relstr{A}$ by its binary projections
to obtain an equivalent~%
(i.e. constant-factor approximability of the problem is preserved)
relational structure.
We replace $\relstr{A}$ with this equivalent structure and,
from this point on,
 assume that all relations in $\relstr{A}$ are binary,
representing each unary relation $U$ by the diagonal binary relation $\{(a,a):a\in U\}$.
It simplifies the presentation of the result:
all the input constraints are binary.
The condition distinguishing the two cases of the dichotomy 
depends on properties of CSP instances over $\relstr{A}$.

An \emph{instance} $\inst I$ of $\CSP(\relstr{A})$ has a variable set $V$, 
and, for each ordered pair $u,v\in V$, a~%
set $\constr{R}_{uv}$ of binary relations called \emph{constraints}. 
We assume that $\constr{R}_{uv}$ and $\constr{R}_{vu}$ 
contain opposite readings of the same relations.
We stress that $\constr{R}_{uu}$ satisfies no additional conditions;
it can contain any binary relations and the opposite readings of relations from $\constr{R}_{uu}$
are in $\constr{R}_{uu}$ as well.
An instance of $\MinCostHom(\relstr{A})$ is an instance of $\CSP(\relstr{A})$
together with a cost function $c_x:A\to\mathbb{Q}_{\geq 0}$ for every variable
$x\in V$.

Note that sometimes,
e.g., when implementing the algorithm in \Cref{sec:tractability},
we work with instances of simpler structure.
However, in other cases we need the definition above 
in its full generality.

An instance $\inst{I}$ is \emph{arc-consistent} if 
    there exist nonempty sets $\{A_u^\inst I\}_{u\in V}$ such that,
    for every pair of variables $u,v$ and every $R\in \constr{R}_{uv}$, 
    one has $\proj_1(R)=A_u^\inst I$ and $\proj_2(R)=A_v^\inst I$.%
    \footnote{Where $\proj_1(R)=\{a\mid (a,b)\in R\}$ 
    and $\proj_2(R)=\{b\mid (a,b)\in R\}$.}
These sets $A_u^\inst I$ are called,
for historical reasons,
\emph{potatoes}.
Let $\inst I$ and $\inst J$ be instances on the same variable set $V$ with specified potatoes. We write $\inst I\leq \inst J$ and say that $\inst I$ is a \emph{subinstance} of $\inst J$ if, for every $u\in V$, $A_u^\inst I\subseteq A_u^\inst J$, and, for all $u,v\in V$,
\[
\constr{R}^{\inst I}_{uv}=\{R\cap (A_u^\inst I\times A_v^\inst I): R\in \constr{R}^{\inst J}_{uv}\}.
\]
Thus, the order $\leq$ always means ``has smaller potatoes and inherited constraints''.

A \emph{pattern} $p$ from $u$ to $v$~%
(in instance $\inst I$) is a sequence of relations $R_1,\dots,R_n$ 
and variables $u=w_0,w_1,\dots,w_n=v$ 
such that $R_i\in \constr{R}^{\inst I}_{w_{i-1}w_i}$ for each $i$. 
The pattern $p$ connects $a$ to $b$ if there exist values $a=a_0,a_1,\dots,a_n=b$ 
such that $(a_{i-1},a_i)\in R_i$ for all $i$. 
If $p$ is a pattern from $u$ to $v$ and $q$ is a pattern from $v$ to $w$,
we write $p+q$ for the concatenation~%
(removing one copy of $v$ at the junction)
of $p$ and $q$, which is a pattern from $u$ to $w$.
For an integer $r\geq 1$, the notation $rp$ denotes the concatenation of $r$ copies of $p$.

An arc-consistent instance is \emph{cycle consistent} if, 
for every variable $u$, every pattern $p$ from $u$ to itself, 
and every $a\in A_u^{\inst I}$, the pattern $p$ connects $a$ to itself.

\begin{definition}
A pair $(\inst I,\inst J)$ is an \emph{ac-leq-cycle pair} if $\inst I$ is arc-consistent, $\inst J$ is cycle consistent, and $\inst I\leq \inst J$.
\end{definition}

\noindent
The following theorem augments \Cref{thm:DichotomyIntro}
by stating the condition distinguishing the two cases of the dichotomy in purely combinatorial terms.
\begin{theorem}\label{thm:HardnessMain}
Let $\relstr{A}$ be a finite relational structure admitting a conservative majority polymorphism. Then exactly one of the following holds:
\begin{itemize}
    \item there exists an ac-leq-cycle pair $(\inst I,\inst J)$ such that $\inst I$ has no solution; in this case a fixed-alphabet cyclic Unique Games deletion problem reduces to $\MinCostHom(\relstr A)$;
    \item for every ac-leq-cycle pair $(\inst I,\inst J)$, the instance $\inst I$ has a solution; in this case $\MinCostHom(\relstr A)$ admits a randomized polynomial-time constant-factor approximation algorithm.
\end{itemize}
\end{theorem}

\subsection{Why the condition is non-algebraic}

The condition in \Cref{thm:HardnessMain} is non-algebraic. 
In fact, no term condition can distinguish the two cases of the dichotomy
even in the presence of majority polymorphisms, as we now show.

The authors of~\cite{DeHaan25:approx} show that the graph $\relstr G$ in~\Cref{fig:counter}
is compatible with a conservative majority operation that returns the top element on every
triple of distinct elements. 
They also prove that $\MinCostHom(\relstr G)$ is not constant-factor
approximable.

\begin{figure}[htb]\begin{center}
\begin{tikzpicture}[vertex/.style={circle, draw, fill=white, minimum size=7mm},edge/.style={thick},scale=1.0]
\node[vertex] (t) at (1,1.6) {};
\node[vertex] (c) at (0,0) {};
\node[vertex] (d) at (2,0) {};
\draw[edge, -] (t) to[out=130, in=50, looseness=7] (t);
\draw[edge] (t) -- (c);
\draw[edge] (t) -- (d);
\draw[edge] (c) -- (d);
\end{tikzpicture}
\end{center}
  \caption{Graph $\relstr{G}$ admitting a conservative majority polymorphism but
  $\MinCostHom(\relstr{G})$ not constant-factor
  approximable.}\label{fig:counter}
\end{figure}

By contrast, a Boolean relational structure defining 2-SAT, $\relstr{A}:=(\{0,1\};R_{00},R_{01},R_{10},R_{11})$~%
(with $R_{ij} = \{0,1\}^2\setminus\{(i,j)\}$)
is compatible with a conservative majority operation, 
and $\MinCostHom(\relstr{A})$ is constant-factor approximable.
However, as observed by Zarathustra Brady, 2-SAT pp-constructs the graph
$\relstr{G}$ from~\Cref{fig:counter}:
let $R_{00}$ be the vertex set of the graph. 
Put an edge between $(x,y)$ and $(x',y')$ whenever $R_{00}(x,x')\wedge R_{00}(y,y')$ holds. 
This defines the same graph, 
with $(1,1)$ as the top element, and closes the pp-construction.

Thus, every term condition satisfied by the polymorphisms $\pol(\relstr{A})$ of 2-SAT is also satisfied by the polymorphisms $\pol(\relstr{G})$ of
the graph $\relstr{G}$. 
Consequently, term conditions, such as being a majority operation or a
near-unanimity operation, cannot distinguish the two cases of our main result.

\section{Hardness}\label{sec:hardness}

The hardness proof uses the following contrapositive principle. If there is an ac-leq-cycle pair $(\inst I,\inst J)$ such that the smaller instance $\inst I$ has no solution, then the excess values in $\inst J\setminus \inst I$ can be used as paid violations. This simulates the deletion version of a fixed-alphabet cyclic Unique Games problem.

\subsection{Proof of hardness in \texorpdfstring{\Cref{thm:HardnessMain}}{the main theorem}}

The hardness will be obtained by a reduction using 
\Cref{thm:AC}
with a proof deferred to the appendix.
The theorem builds on two notions: 
a \emph{directed cycle}
is defined in \Cref{sec:intro} as  $Z_k=\{(a,a+1 \bmod k):a\in\{0,\dots,k-1\}\}$~%
(in this section we allow $k=1$,
which makes $Z_1$ a digraph with one vertex and a self-loop),
and a union of directed cycles is a disjoint union of some directed cycles.
A \emph{grid} is a directed graph on $\{0,\dotsc,k-1\}^2$ with an arc $((a,b),(c,d))$ if and only if $b=c$.
Finally, a \emph{blow-up} of a digraph is obtained by replacing each vertex 
by a non-empty set 
and each arc by a complete directed bipartite graph.

\input{key-structures-figure}

The main tool of the reduction will be the following combinatorial result about arc-consistent instances.

\begin{restatable}{theorem}{AC}\label{thm:AC}
Assume that $\inst I$ is binary, arc-consistent
and has at least one potato with at least two elements. Then:
\begin{enumerate}
    \item\label{it:ac-subinstance} there exists a proper~%
    (i.e., at least one $A_u^{\inst I'}\varsubsetneq A_u^{\inst I}$) 
    arc-consistent subinstance $\inst I' \leq \inst I$, or
    \item\label{it:ac-cycle} there is a variable $x$ and a pattern $p$ from $x$ to $x$,
    such that the ``connected by $p$'' relation on $A_x^\inst I$ 
    defines 
    a blow-up of a disjoint union of directed cycles, at least one of length $\geq 2$, or
    a blow-up of a grid of size $\geq 2$, or 
    \item\label{it:ac-linked} there exists a variable $x$, two elements $0,1\in A_x^\inst I$, 
    and patterns $p_{00}$ and $p_{11}$ from $x$ to $x$ such that, 
    after restricting intermediate potatoes to suitable two-element subsets, 
    $p_{00}$ defines $R_{00}$ or $\neq$ and $p_{11}$ defines $R_{11}$ or
    $\neq$.
\end{enumerate}
\end{restatable}

\begin{proof}[Proof of hardness in \Cref{thm:HardnessMain}]
Assume there exists an ac-leq-cycle pair $(\inst I,\inst J)$ such that $\inst I$ has no solution. Choose such a pair with $\sum_u |A_u^\inst I|$ minimal. Since $\inst I$ is arc-consistent, if every potato were a singleton, the singleton assignment would be a solution. 
Hence, some potato has size at least two,
and we apply~\Cref{thm:AC} to $\inst I$.

\Cref{it:ac-subinstance} cannot occur,
as this would give a proper arc-consistent subinstance $\inst I'\leq\inst I$ and then $(\inst I',\inst J)$ would be an ac-leq-cycle pair. 
Moreover, any solution to $\inst I'$ is a solution to $\inst I$ so $\inst I'$ needs to be unsolvable -- 
this contradicts the minimality of $\inst I$. 
We are left with \Cref{it:ac-cycle,it:ac-linked}.
In these two cases we provide a reduction,
using gadgets.

The gadgets are built from patterns in instance $\inst I$ and are called \emph{gadget instances}.
Let $q=R_1,\dots,R_n$ be a pattern along variables $x=w_0,w_1,\dots,w_n=x$. 
Note that variables and constraints in a pattern can repeat~%
(as does $x$ for example). 
To construct the gadget instance we make all the copies distinct;
the gadget instance has exactly $n+1$ variables and $n$ constraints.
Each solution to the gadget corresponds to a valuation witnessing the ``connected by $q$'' relation.
Note that we have two instances $\inst I$ and $\inst J$ on the same variable set;
we will use this fact by (usually) making the evaluations in $\inst I$ have cost $0$ and the ones in $\inst J\setminus\inst I$ have cost at least $1$.
In the reduction, the gadget constraints are the corresponding relations of $\inst J$, with the domain restrictions and unary costs specified below.

Consider the case of \Cref{it:ac-cycle} with a blow-up of a union of directed cycles.
Let $x$ be the variable, let $p$ be the pattern, and let $C=B_0\cup\cdots\cup
  B_{k-1}\subseteq A_x^\inst I$ be the vertices of a blow-up of a directed cycle, with $k\geq 2$. 
Thus, in $\inst I$, the pattern $p$ connects precisely the pairs $B_i\times B_{i+1}$~(for every $i$ with addition modulo $k$),
and no other pairs inside $C$.
Because $\inst I\leq\inst J$, every connection realized by $p$ in $\inst I$ is also realized by $p$ in $\inst J$. 
Moreover, 
since $\inst J$ is cycle consistent, 
every closed pattern from $x$ to $x$ connects every element of $A_x^\inst J$ to itself in $\inst J$.

\begin{namedlemma}[Saturation on cycle blow-ups]\label{lem:saturation}
Let $q=(k+1)p$ be the concatenation of $k+1$ copies of $p$. 
Then, in $\inst I$, the pattern $q$ defines the same relation on $C$ as $p$, and, in $\inst J$, the pattern $q$ connects every element of $C$ to every element of $C$.
\end{namedlemma}

\begin{proof}
For the first claim: $p$ advances by one class modulo $k$, 
so $(k+1)p$ advances by $k+1\equiv 1\pmod k$.

For the second, fix $a\in B_i$ and $b\in B_j$. 
Choose $r\in\{1,\dots,k\}$ such
  that $i+r\equiv j\pmod k$. 
In $\inst I$, and hence in $\inst J$, 
  $r$ copies of
  $p$ can advance from $a$ to any element of $B_j$; 
  choose this element to be $b$. 
On the remaining $k+1-r$ copies of $p$, use cycle consistency of $\inst J$: 
  the closed pattern $p$ connects each current value to itself. 
Hence, $(k+1)p$ connects $a$ to $b$ in $\inst J$.
\end{proof}

To finish this case we reduce from $\cUGC(k)$:
given an instance of $\cUGC(k)$,
for every variable $y$ of the $\cUGC(k)$ instance, 
create a variable $y$ of instance $\inst K$ and set 
$A_y^\inst K=C$; the values $A\setminus C$ are forbidden by conservativity.
For every constraint $Z_k(y,z)$ of the input instance,
 insert, into $\inst K$, a fresh copy of the gadget instance for pattern $q$ 
 from \namedlemcref{lem:saturation}
 and identify its endpoints with $y$ and $z$. 
Inside the gadget give cost $0$ to values of each $w_i$ 
in the smaller potatoes $A_{w_i}^\inst I$ 
and cost $1$ to values in $A_{w_i}^\inst J\setminus A_{w_i}^\inst I$. 
The construction is complete and we now verify its correctness.

The correspondence between evaluations of the variables in an instance of $\cUGC(k)$
 and evaluations of the variables in the constructed instance is clear:
 for every $j\in\{0,\dots,k-1\}$, variable $y$ maps to $j$~%
 (in the input instance)
 if and only if 
 $y$ maps to (any value in) $B_j$~%
 (in $\inst K$).
The values $y,z$ satisfy $Z_k(y,z)$ 
if and only if 
we can assign values, in $\inst K$, to the gadget using only the elements of $\inst I$ 
-- this incurs cost $0$.
For any other pair of values,
we use \namedlemcref{lem:saturation} to evaluate the gadget inside $\inst J$,
which incurs a cost of at least $1$ on some variable as the connection 
needs to exit $\inst I$.
Thus, the reduction is approximation-preserving, up to the fixed gadget constant.
This finishes the reduction in the case of a blow-up of a union of directed cycles.

Now, working in \Cref{it:ac-cycle} with a blow-up of a grid,
let $x$ be a variable with $A_x^\inst I$ split into $B_{ij}$ for $i,j\in\{0,\dots,k-1\}$, 
and let $p$ be the pattern from $x$ to $x$ defining a blow-up of a grid with classes $B_{ij}$.
Let $q=3p$; we reduce from $\cUGC(2)$;
take an instance of such a problem.
For every variable $y$ of the $\cUGC(2)$ instance, 
create a variable $y$ of instance $\inst K$ and set 
$A_y^\inst K = B_{01}\cup B_{10}$; other values are forbidden by conservativity.
For every constraint $Z_2(y,z)$ of the input instance,
 insert into $\inst K$ a fresh copy of the gadget instance for pattern $q$ 
 and identify its endpoints with $y$ and $z$. 
Inside the gadget give cost $0$ to values of each $w_i$ 
in the smaller potatoes $A_{w_i}^\inst I$,
except on the copies of a variable $x$ at the junction of $p$ inside $p+p+p$;
there, cost $0$ is assigned to elements in $B_{01}\cup B_{10}$.  
Set cost $1$ to all the remaining values.

The analysis is identical to that in the previous case:
any value from $B_{01}\cup B_{10}$ can be connected to any value from $B_{01}\cup B_{10}$ 
in $\inst J$ using the pattern $q$~%
(here we use $q=3p$ just like in the proof of \namedlemcref{lem:saturation}).
The connections of cost $0$ must use only values from $\inst I$ and,
on the junction, elements from $B_{01}\cup B_{10}$.
Thus, the only connections of cost $0$ are from $B_{01}$ to $B_{10}$ and from $B_{10}$ to $B_{01}$.
This proves the correctness of the reduction in the case of a blow-up of a grid.

It remains to treat \Cref{it:ac-linked}. 
We have a variable $x$, two values $0,1\in A_x^\inst I$, and two closed patterns $p_{00}$ and $p_{11}$. 
We construct the reduction from $\cUGC(2)$, as in the previous cases, but using both patterns $p_{00}$ and $p_{11}$ to build two gadgets for each constraint.
Moreover, we will be assigning cost $0$ only to the values from the ``restricted two-element subsets'' of the potatoes along the patterns, and cost $1$ to all other values in the potatoes in $\inst I$ as well as in $\inst J$.

Formally, we reduce from $\cUGC(2)$ 
and for every input variable $y$, create a variable forced, using conservativity, to take one of the two values $0,1$. 
For every constraint $Z_2(y,z)$, insert two fresh, disjoint gadgets with endpoints identified with $y$ and $z$: 
one obtained from $p_{00}$ and one from $p_{11}$. 
In both gadgets give cost $0$ to the values in the distinguished two-element sets assigned to variables of the gadgets. 
Set cost $1$ to all other values of the corresponding larger potatoes,
i.e. the potatoes coming from $\inst J$.

Assigning the pair of variables $(y,z)$ the values $(0,1)$ or $(1,0)$ 
can be accomplished by realizing $p_{00}$ and $p_{11}$ in
the selected two-element sets,
so both gadgets have cost $0$. 
For endpoint values $(0,0)$, the $p_{00}$ gadget has no all-zero-cost realization,
hence any realization~%
(and such realizations exist by cycle consistency of $\inst J$ exactly as in the previous cases)
must leave the distinguished two-element sets and has cost at least $1$. 
The $p_{11}$ gadget can be realized by cycle consistency as well,
although its cost is not determined.
Similarly, for endpoint values $(1,1)$, the $p_{11}$ gadget has cost at least $1$.

Consequently, the constructed instance has cost $0$ exactly on assignments satisfying the corresponding disequality constraints, 
and every violated disequality constraint contributes cost at least $1$, up to the fixed size of the two gadgets. 
This is an approximation-preserving reduction from the fixed-alphabet deletion problem for $Z_2$.
This proves the hardness part of~\Cref{thm:HardnessMain}.
\end{proof}

\section{Tractability}\label{sec:tractability}

Consider an instance of $\MinCostHom(\relstr A)$ with a variable set $V$ and cost functions
$c_x:A\to\mathbb{Q}_{\geq 0}$, $x\in V$.  The tractability argument has three
steps.  First, we preprocess the instance to obtain a cycle-consistent ambient
instance.  Next, we solve the basic linear programming (BLP) relaxation
(cf.~\Cref{sec:tractability-rounding-cost} for details) 
and use a randomized threshold rounding
to define a subinstance.  Finally, we prove that this truncated instance is always
solvable, while any of its solutions has expected cost within a constant factor of the
LP value.

Since all costs are unary, all binary constraints between the same ordered pair
of variables may be conjoined.  If there was no original constraint between $x$
and $y$, we put $R_{xy}=A\times A$.  Thus, throughout this section, for every
ordered pair of variables $(x,y)$ there is a single binary relation $R_{xy}$.
We can further assume that $R_{yx} = - R_{xy}$~%
(where $-R$ is the opposite reading of $R$).
For $R_{xx}$ we go even further:
we intersect all the constraints on $(x,x)$ and 
then restrict the values to equal pairs:
in the end we obtain
$
        R_{xx}=\{(a,a):a\in A_x^{\inst I}\}
$
for suitably chosen $A_x^{\inst I}$.
This changes neither feasibility nor the unary objective.

We first enforce \emph{$(2,3)$-consistency}.  Concretely, we repeatedly delete tuples
until, for every triple of variables $x,y,z$ and every $(a,b)\in R_{xy}$, there is
$c\in A_z$ such that
\[ (a,c)\in R_{xz} \quad\text{and}\quad (b,c)\in R_{yz}. \]

The potatoes are the projections of the remaining binary relations.  If some potato
becomes empty, the original instance has no feasible solution;
in this case the algorithm outputs ``infeasible'' and stops.
Otherwise, denote the resulting instance by $\inst J$.
Passing from the original instance to $\inst J$ preserves the set of solutions,
and \Cref{prop:23-cycle-consistency-levels} in the appendix shows that
$\inst J$ is arc-consistent and cycle-consistent~%
(which is well-known, but we include a proof for completeness).

Now solve the BLP relaxation on $\inst J$, with costs inherited from the
original instance.  Write $\lambda_v(a)$ for the LP marginal of the assignment
$v\mapsto a$, and put
\[
        \lambda_v(B):=\sum_{b\in B}\lambda_v(b)
        \qquad (B\subseteq A_v^{\inst J}).
\]
Since preprocessing is solution-preserving and constraints are strict,
$ \mathrm{BLP}(\inst J)\leq \mathrm{OPT}$.

Next we choose one global threshold sequence $(t_i)_{i\geq 1}$ for the whole
instance by sampling independently for each $i\geq 1$:
\[
        t_i\sim \mathrm{Unif}\bigl(2^{-i},2^{-(i-1)}\bigr].
\]
The same threshold $t_i$ is used for every variable.

\begin{definition}[Threshold conflict]
The assignment $v\mapsto a$ conflicts with $t_i$ if there is
$B\subseteq A_v^{\inst J}$, $a\notin B$, such that
\[
        \lambda_v(B)<t_i\leq \lambda_v(B\cup\{a\}).
\]
\end{definition}

\begin{definition}[Good level]
  A level $i$ is \emph{good for $v$} if there is no $a\in A_v^{\inst J}$ such that
\[
        \lambda_v(a)<t_i
\]
and $v\mapsto a$ conflicts with one of $t_1,\ldots,t_i$.
\end{definition}

For every value $a\in A_v^{\inst J}$, define
\[
        \ell(v,a)=i
        \quad\Longleftrightarrow\quad
        \lambda_v(a)\in (2^{-i},2^{-(i-1)}],
\]
with $\ell(v,a)=\infty$ if $\lambda_v(a)=0$.  
Thus~%
(except in the $\infty$ case)
\[
        2^{-\ell(v,a)}<\lambda_v(a)\leq 2^{-(\ell(v,a)-1)}.
\]

\noindent 
Using the same global threshold sequence, compute levels of variables, $L(v)$,
which will form the basis for truncating the instance.
For each variable $v$, run the following procedure:
\begin{enumerate}
    \item Set $\mathrm{L}:=1$.
    \item While $\mathrm{L}$ is not good for $v$
        choose $a\in A_v^{\inst J}$
    such that
    \begin{itemize} 
        \item $\lambda_v(a)<t_{\mathrm{L}}$, and 
        \item $v\mapsto a$ conflicts with one of
        $t_{1},\ldots,t_{\mathrm{L}}$, and
        \item it has minimal $\lambda_v(a)$ among such witnesses.
    \end{itemize}
    Then set $\mathrm{L}:=\ell(v,a)+1$.
    \item Set $L(v):=\mathrm{L}$.
\end{enumerate}
Define the truncated potato by the random cutoff
\[
        A_v^{\inst I}:=\{a\in A_v^{\inst J}:\lambda_v(a)\geq t_{L(v)}\}.
\]
Let $\inst I$ be the instance inherited from $\inst J$ by these potatoes.  
The rest of the argument rests on two facts.
\begin{restatable}{proposition}{propRoundedInstanceApproximation}\label{prop:rounded-instance-is-arc-consistent}
If the truncated instance $\inst I$ has a solution, then taking any solution of
$\inst I$ gives a randomized constant-factor approximation in expectation.
\end{restatable}
\begin{restatable}{proposition}{propRoundedInstanceSolvable}\label{prop:rounded-instance-has-solution}
The truncated instance $\inst I$  has a solution.
\end{restatable}
 With these two facts in hand, we can run any algorithm for finding a solution to
$\inst I$.  Since $\inst I$ is a CSP instance with a majority polymorphism,
it has bounded width and thus local consistency algorithms solve $\inst I$ in
polynomial time, e.g.~\cite{Feder98:sicomp,Kozik21:sicomp}. Returning any solution gives a randomized
constant-factor approximation in expectation.

%% file: key-structures-figure.tex
\begin{figure}[ht]
\centering
\tikzset{
  gv/.style={circle, draw, thick, fill=white, minimum size=9mm, inner sep=1pt},
  ge/.style={->, thick, >=Stealth}
}
\begin{subfigure}[b]{0.45\textwidth}
\centering
\begin{tikzpicture}[scale=0.85]
\useasboundingbox (-0.5,-1.5) rectangle (3.0,3.5);
\node[gv] (a0) at (0.25,2.6) {$0$};
\node[gv] (a1) at (2.25,2.6) {$1$};
\draw[ge] (a0) to[bend left=30] (a1);
\draw[ge] (a1) to[bend left=30] (a0);
\node[gv] (b0) at (1.25,0.73) {$0$};
\node[gv] (b1) at (2.25,-1.0) {$1$};
\node[gv] (b2) at (0.25,-1.0) {$2$};
\draw[ge] (b0) -- (b1);
\draw[ge] (b1) -- (b2);
\draw[ge] (b2) -- (b0);
\end{tikzpicture}
\caption{Disjoint union $C_2 \sqcup C_3$}
\label{fig:cycles}
\end{subfigure}
\hspace{-0.10\textwidth}
\begin{subfigure}[b]{0.45\textwidth}
\centering
\begin{tikzpicture}[scale=0.85]
\useasboundingbox (-0.5,-1.5) rectangle (3.0,3.5);
\node[gv] (00) at (0.2,2.0) {$00$};
\node[gv] (01) at (2.3,2.0) {$01$};
\node[gv] (10) at (0.2,0.1) {$10$};
\node[gv] (11) at (2.3,0.1) {$11$};
\draw[ge] (00) to[out=120,in=60,looseness=5] (00);
\draw[ge] (11) to[out=300,in=240,looseness=5] (11);
\draw[ge] (00) -- (01);
\draw[ge] (01) -- (11);
\draw[ge] (11) -- (10);
\draw[ge] (10) -- (00);
\draw[ge] (01) to[bend left=20] (10);
\draw[ge] (10) to[bend left=20] (01);
\end{tikzpicture}
\caption{Grid on $\{0,1\}^2$}
\label{fig:grid}
\end{subfigure}
\caption{Examples of the two key structures used in \Cref{thm:AC}: a disjoint union of directed cycles (left) and the grid digraph on $\{0,1\}^2$ (right).}
\label{fig:key-structures}
\end{figure}

%% file: ACproof.tex
\section{Proof of~\texorpdfstring{\cref{thm:AC}}{AC theorem}}
\label{sec:AC}

\noindent We recall the statement of the theorem we are proving. It might be
worth noting that, in the statement of the theorem below, the instance $\inst
I'$ and each blown-up vertex of the grid or union of directed cycles are
pp-definable from $\inst I$ and constants. Recall that a relation $R\subseteq
A^m$ is primitive positive definable (or pp-definable, for short) from a set of relations $\mathcal{R}$ on $A$ if $R$ can be
defined by a first-order formula that uses only relations from $\mathcal{R}$,
the equality relation, conjunction, and existential quantification. 
\AC*
\noindent The first observation allows us to deal only with \emph{connected inputs}, i.e.,
instances where every two variables are connected by a pattern.
\begin{observation}
    If~\Cref{thm:AC} holds for instances with connected input,
    then it holds for all instances.
\end{observation}
\begin{proof}
    If $\inst I$ has disconnected input we choose a variable $x$ with a potato of size at least two,
    and build the instance $\inst J$ by removing from $\inst I$ all variables not connected to $x$ by a pattern,
    as well as all constraints on them.
    Then $\inst J$ is connected and has at least one potato of size at least two and we can apply \Cref{thm:AC} to $\inst J$.
    In \Cref{it:ac-cycle,it:ac-linked} all the pattern evaluations in $\inst J$ are identical to those in $\inst I$ and we are done.
    In \Cref{it:ac-subinstance} we get a proper subinstance $\inst J'$ of $\inst J$ and we can extend it to a proper subinstance $\inst I'$ of $\inst I$ by taking $A_y^{\inst I'} = A_y^{\inst I}$ for all variables $y$ not in $\inst J$. 
\end{proof}
In this section we,
most often,
work within a fixed instance,
usually $\inst I$ from the statement of \Cref{thm:AC}.
Recall that,
for two patterns $p$ and $q$ within the same instance
we write $p+q$ for concatenation of patterns and $-p$ for the reverse pattern.
If $p$ is a pattern from $x$ to $y$,
then
\[
    \Conn{p} = \{(a,b): \text{$a$ is connected to $b$ by pattern $p$}\},
\]
\noindent 
Furthermore, if $B\subseteq A_x^{\inst I}$, then
\[
B+p= B+\Conn{p} = \{b\in A_y^{\inst I}:\text{ some }a\in B\text{ is connected to }b\text{ by }p\}.
\]
\noindent We need three standard definitions:
a relation $R\subseteq B\times C$ is:
\begin{itemize}
    \item \emph{proper} if $\emptyset \neq R \neq B\times C$,
    \item \emph{subdirect} if $\proj_1(R)=B$ while $\proj_2(R) =C$,
    and
    \item \emph{linked} if it is subdirect and the transitive closure of $R'$ is $B^2$ where 
    \[
            R'(x,y) = \exists v\  R(x,v)\wedge R(y,v).
    \] 
\end{itemize}
The following theorem establishes a part of the claim:
\Cref{it:binary-subinstance,it:binary-cycle} are identical to \Cref{it:ac-subinstance,it:ac-cycle}.
\begin{restatable}{theorem}{BinaryCycleOrSubinstance}\label{thm:binary-cycle-or-subinstance}
Assume that $\inst I$ is binary,
arc-consistent and has at least one potato of size at least two.
Moreover, assume that $\inst I$ has connected input.
Then at least one of the following holds:
\begin{enumerate}
    \item\label{it:binary-subinstance} there exists a proper~%
    (i.e., at least one $A_u^{\inst I'}\varsubsetneq A_u^{\inst I}$) 
    arc-consistent subinstance $\inst I' \leq \inst I$,
    \item\label{it:binary-cycle} there is a variable $x$ and a pattern $p$ from $x$ to $x$,
    such that the ``connected by $p$'' relation on $A_x^\inst I$ 
    defines a blow-up of a grid of size $\geq 2$,
    or a blow-up of a disjoint union of directed cycles, at least one of length $\geq 2$,
    \item\label{it:binary-linked} there is a pattern defining a proper linked subrelation of a product of two potatoes.
\end{enumerate}
\end{restatable}

\noindent The linked-pattern case in \Cref{thm:binary-cycle-or-subinstance} does not immediately give the desired conclusion.
We handle it using a binary polymorphism of the instance.
\begin{definition}
Let $\inst I$ be a binary instance.
A \emph{binary instance polymorphism},
or simply a \emph{binary polymorphism},
is a family of operations
\[
\{f_x:A_x^2\to A_x\}_x
\]
such that,
for every $x$ and $y$ and every  $R\in\constr{R}_{xy}$,
whenever
\[
(a,b),(a',b')\in R,
\]
one has
\[
(f_x(a,a'),f_y(b,b'))\in R.
\]
The polymorphism is \emph{conservative} if every $f_x$ is conservative,
and \emph{commutative} if every $f_x$ is commutative.

Moreover, we say that a nonempty subset $B\subseteq A_x$ \emph{2-absorbs} $A_x$ if
\[
f_x(B,A_x)\cup f_x(A_x,B)\subseteq B.
\]
\end{definition}
This concept is very useful; it allows us to turn the linked-case
into what we need in the proof of~\Cref{thm:binary-cycle-or-subinstance}. 
Given a binary instance with a commutative conservative binary polymorphism,
we can turn \Cref{it:binary-linked} of~\Cref{thm:binary-cycle-or-subinstance}
into \Cref{it:ac-subinstance} of~\Cref{thm:binary-cycle-or-subinstance} by the following two propositions.

\begin{restatable}{proposition}{propLinkedGivesAbsorption}\label{prop:linked-gives-absorption-short}
    Let $\inst I$ be a binary instance with a conservative commutative binary polymorphism $\{f_x\}_x$.
    If some pattern defines a proper linked subrelation of $A_x^{\inst I}\times A_y^{\inst I}$,
    then either $A_x^{\inst I}$ or $A_y^{\inst I}$ has a proper 2-absorbing subset.
\end{restatable}

\begin{restatable}{proposition}{propAbsorptionGivesSubinstance}\label{prop:absorption-gives-subinstance-short}
Let $\inst I$ be a binary arc-consistent instance with a conservative commutative binary polymorphism $\{f_x\}_x$.
If there are $x$ and $B\varsubsetneq A_x$ such that $B$ 2-absorbs $A_x$,
then $\inst I$ has a proper arc-consistent subinstance.
\end{restatable}

It remains to provide the commutative conservative binary polymorphism.
We associate with every binary instance a directed graph which controls the existence of conservative commutative binary polymorphisms.
\begin{definition}
Let $\inst I$ be a binary instance.
The directed graph $G(\inst I)$ is defined as follows:
\begin{itemize}
    \item the vertices are triples $(a,b;x)$ with $a\neq b\in A_x$;
    \item there is an arc
    \[
    (a,b;x)\longrightarrow (c,d;y)
    \]
    if and only if there exists $R\in\constr{R}_{xy}$ such that
    \[
    (a,c),(b,d)\in R
    \qquad\text{and}\qquad
    (a,d)\notin R.
    \]
\end{itemize}
The vertices $(a,b;x)$ and $(b,a;x)$ are called \emph{dual}.
\end{definition}
\noindent The following theorem ties the two concepts together.

\begin{restatable}{theorem}{thmBinaryPolymorphismObstruction}\label{thm:binary-polymorphism-obstruction}
A binary instance $\inst I$ has a conservative commutative binary polymorphism 
if and only if no dual vertices belong to the same strongly connected component of $G(\inst I)$.
\end{restatable}

\noindent It remains to turn two dual vertices of $G(\inst I)$ into \Cref{it:ac-linked}.
Let the vertices be $(0,1;x)$ and $(1,0;x)$;
the two paths in $G(\inst I)$ connecting them give two patterns $p$ and $q$ from $x$ to $x$ 
and the two-element restrictions.
It follows that, restricting patterns at each step, $(0,1),(1,0)\in\Conn{p}$ while $(0,0)\notin\Conn{p}$
and similarly $(0,1),(1,0)\in\Conn{q}$ while $(1,1)\notin\Conn{q}$.
We obtain relations $R_{00}$ or $\neq$ in the first case and $R_{11}$ or $\neq$ in the second case.
This proves \Cref{it:ac-linked}.

\subsection{Proof of \texorpdfstring{\Cref{thm:binary-cycle-or-subinstance}}{the previous theorem}}
    Let us recall the theorem we are proving.
    \BinaryCycleOrSubinstance*

    The remainder of this section is devoted to the proof.
    First we consider the case of $\inst I$ not connected~%
    (as a microstructure,
    i.e. the graph whose vertex set is the disjoint union $\bigsqcup_x A_x^{\inst I}$ 
    and edges defined by the relations in constraints of $\inst I$).
    In such a case every $A_x^{\inst I}$ has at least two elements.

    Fix any $x$ and consider all reflexive relations on $A_x$ defined by patterns from $x$ to $x$;
    let $q$ be the sum of all these patterns.%
    \footnote{The sum is only potentially infinite,
    as for each relation realized by a pattern we choose 
    a single representative.}
    Note that $\Conn{q}$ is reflexive and $\Conn{q}\supseteq \Conn{p}$ for every $p$ defining a reflexive relation on $A_x$.
    Note that $-\Conn{q}$ is reflexive and represented by some pattern~%
    (pattern $-q$ for example);
    consequently $-\Conn{q}=\Conn{-q}\subseteq \Conn{q}$,
    i.e. $\Conn{q}$ is symmetric.
    By the same token $\Conn{q}+\Conn{q} = \Conn{q+q}\subseteq \Conn{q}$,
    i.e. $\Conn{q}$ is transitive.
    We have concluded that $\Conn{q}$ is,
    in fact,
    an equivalence.
   
    If at least one class of $\Conn{q}$ corresponds to a connected component of $\inst I$ we have a reduction:
    the class is pp-definable from $\inst I$ and constants,
    and for every $y$ we can define $A_y^{\inst I'}$ by connecting it via a pattern to this class.
    We obtain a subinstance and \Cref{it:binary-subinstance} holds.
   
    If no class of $\Conn{q}$ coincides with a connected component, then
    we fix two $a,b\in A_x^{\inst I}$ such that $a$ and $b$ are in the same connected component of $\inst I$ but are not $\Conn{q}$ related.
    By connectivity one can find a pattern $q'$ such that $a$ is connected to $b$ by $q'$.
    We will show that the pattern $q+q'$ can be taken for $p$ to satisfy \Cref{it:binary-cycle},
    with the $\Conn{q}$ classes as the blow-up blocks.
    It suffices to show that,
    for any 
    $c\in A_x^{\inst I}$ the set $\{c\}+q+q'$ is the $\Conn{q}$ class,
    and that if $c$ and $c'$ are in the same $\Conn{q}$ class 
    then $\{c\}+q+q' = \{c'\}+q+q'$.
    
    Let $C= \{c\}+q+q'$; 
    the set $C$ is a subset of a $\Conn{q}$ class
    as $\Conn{-q'-q+q+q'}$ is reflexive and connects every element of $C$ 
    to every element of $C$. 
    Note that if the class were bigger, say $C+q = D\varsupsetneq C$, 
    then $\{c\}+q \varsubsetneq \{c\}+q+q'+q-q' = C+q-q'= D-q'$ which shows~%
    (using arc consistency) that 
    $\{c\}+q$ is not a $\Conn{q}$ class, a contradiction.
    As $\{c'\} +q = \{c\}+q$ for any $c'$ in the same $\Conn{q}$ class as $c$,
    the second condition follows.
    Thus, we have a blow-up of a disjoint union of cycles,
    and the cycle including $a$ and $b$ has length at least $2$.

    We are now in the case that $\inst I$ has one connected component.
    Fix $x$ such that $|A^{\inst I}_x|>1$, 
    and let $q$ be a closed pattern from $x$ to $x$ such that $\Conn{q}\neq (A_x^{\inst I})^2$  is an equivalence relation.
    If no such pattern exists, 
    we take $q$ to be the empty pattern and let the equivalence $\Conn{q}$ be the identity.
    If more than one equivalence exists, we require that this equivalence is maximal,
    under inclusion,
    among equivalence relations defined by closed patterns at $x$.
    Call this equivalence $\theta_x$.

    Let $r$ be a pattern from $x$ to a variable $y$.
    We claim~%
    (and prove in the next paragraph) 
    that unless \Cref{it:binary-linked} holds, exactly one of the following happens:
    \begin{itemize}
        \item for every $a\in A^{\inst I}_x$,
        one has
        \[
        a/\theta_x+r=A^{\inst I}_y;
        \]
        \item the sets $a/\theta_x+r$,
        for $a\in A^{\inst I}_x$,
        form classes of an equivalence relation;
        in this case we call $r$ a \emph{transferring pattern} 
        and denote the equivalence it defines by $\theta_y$.
    \end{itemize}
    
    Indeed,
    consider $\Conn{q+r}$;
    if it is full, then the first case above holds.
    If it is not full and linked, then we are in \Cref{it:binary-linked}.
    Hence,
    for some $k$, the relation
    $\Conn{k(q+r-r-q)}$
    is a proper equivalence on $A^{\inst I}_x$ and it includes $\theta_x$.
    By maximality of $\theta_x$,
    $\Conn{k(q+r-r-q)}$ is equal to $\theta_x$.
    It follows that the sets $a/\theta_x+r$ are precisely the classes 
    of an equivalence relation on $A^{\inst I}_y$,
    as claimed.

    We next show that two different patterns from $x$ to the same variable $y$
    cannot define two different equivalence relations, 
    unless \Cref{it:binary-cycle} or \Cref{it:binary-linked} holds.
    Suppose that $r,r'$ are patterns from $x$ to $y$ 
    defining distinct equivalence relations $\theta_y$ and $\theta'_y$ 
    on $A^{\inst I}_y$.
    Note that $\Conn{r}$ establishes a bijection 
    between the $\theta_x$-classes and the $\theta_y$-classes,
    similarly $\Conn{r '}$ establishes a bijection 
    between the $\theta_x$-classes and the $\theta'_y$-classes.
    Consider the pattern $-r+q+r'$ from $y$ to $y$.
    The relation $R=\Conn{-r+q+r'}$ establishes a bijection between the $\theta_y$-classes and the $\theta'_y$-classes
    and, moreover, is full within each pair of classes.

    Consider the relation $S=R+R$~%
    (i.e. the relation defined by the pattern $-r+q+r'-r+q+r'$),
    and then $T=S-S+S+\dotsb -S$ sufficiently many times
    to make it transitive.
    The relation $T$ is an equivalence relation on $A^{\inst I}_y$,
    and it cannot be $\theta_y$ because 
    there exists a class of $\theta_y'$ which intersects at least two classes of $\theta_y$.
    It follows that for some $a\in A^{\inst I}_y$, $a+S$ contains 
    at least two classes of $\theta'_y$
    and consequently $a+S-S$ contains at least two classes of $\theta_y$.
    Thus, by maximality of $\theta_y$~%
    (which follows from maximality of $\theta_x$ as the transfer by $-r$ 
    would provide a bigger equivalence on $A^{\inst I}_x$ from a bigger equivalence on $A^{\inst I}_y$),
    the equivalence is full which implies that either we have~\Cref{it:binary-linked}~%
    (and we are done with the proof)
    or $S$ is full.

    Since $S$ is full, we conclude that each class of $\theta_y$ 
    intersects each class of $\theta'_y$:
    indeed, for $a\in A^{\inst I}_y$,
    $B=a+R$ is a $\theta'_y$-class corresponding to the $\theta_y$-class of $a$,
    and $a+R+R=B+R$ is the union of $\theta'_y$-classes 
    corresponding to the $\theta_y$ classes which intersect $B$.

    Let $k$ be the number of $\theta_y$-classes,
    which we enumerate by $0,\dotsc,k-1$.
    We define $B_{ij}$ to be the intersection of $i$-th $\theta_y$ class
    with the class of $\theta'_y$ corresponding to the $j$-th $\theta_y$-class.
    It follows that $(a,b)\in R$ 
    if and only if $a\in B_{ij}$ and $b\in B_{li}$ for some $i,j,l$.
    Thus, $R$ defines a blow-up of a grid,
    and we are in \Cref{it:binary-cycle}.
    We are left with the case of unique $\theta_y$ for every $y$~%
    (which can be full, in case there is no transferring pattern from $x$ to $y$).

    We know that,
    for every variable $y$,
    all patterns from $x$ to $y$ which transfer $\theta_x$ 
    define the same equivalence relation $\theta_y$.
    We now rule out the possibility that two such patterns 
    induce different permutations of the same equivalence classes~%
    (unless \Cref{it:binary-cycle} holds).
    Let $r,r'$ be patterns from $x$ to $y$ defining the same equivalence relation $\theta_y$,
    but suppose that,
    for some $a\in A^{\inst I}_x$,
    \[
    a/\theta_x+r\neq a/\theta_x+r'.
    \]
    Let $R=\Conn{-r+q+r'}$, a binary relation on $A_y^{\inst I}$;
    for any $a\in A^{\inst I}_y$ the set $\{a\}-r$ is a subset of a $\theta_x$ class 
    associated with $a/\theta_y$ via pattern $r$, and $\{a\}-r+q$ is this whole class.
    Finally, $\{a\}-r+q+r'$ is a full $\theta_y$ class.

    Therefore, $R$, viewed as a digraph, is a blow-up of a disjoint union of directed cycles,
    with the blow-up blocks being the $\theta_y$-classes.
    Moreover, the cycle corresponding to the block $(a/\theta_x)+r$ has length at least $2$.
    This shows that \Cref{it:binary-cycle} holds.

    In the remaining case if two patterns from $x$ to $y$ transfer $\theta_x$,
    then they define the same equivalence relation $\theta_y$ and the same permutation of $\theta_y$-classes.
    In this case we construct a proper subinstance and conclude the proof in the case~\Cref{it:binary-subinstance}.
    
    Fix one $\theta_x$-class and call it $A^{\inst I'}_x$.
    For every variable $y$ reachable from $x$ by a transferring pattern $r$
    define
    \[
    A^{\inst I'}_y=A^{\inst I'}_x+r,
    \]
    By the preceding paragraphs, this definition is independent of the choice of $r$.
    For variables not reachable in this way,
    put $A^{\inst I'}_y=A^{\inst I}_y$.
    For every pair of variables $yz$ restrict $\constr{R}_{yz}$,
    by intersecting each relation with $A^{\inst I'}_y\times A^{\inst I'}_z$.
    It remains to prove that $\inst I'$ is arc-consistent.
    
    Take any $y$ and $z$;
    if $A_y^{\inst I'}=A_y^{\inst I}$ and $A_z^{\inst I'}=A_z^{\inst I}$
    the relations in $\constr{R}_{yz}$ stay the same so the condition holds.
    If $A_y^{\inst I'}\neq A_y^{\inst I}$ and $A_z^{\inst I'}=A_z^{\inst I}$,
    take any $R\in\constr{R}_{yz}$ and a transferring pattern $r$ from $x$ to $y$.
    Extend the pattern $r$ by a step to $z$ using $R$ and call the resulting pattern $r'$.
    We must have $a/\theta_x+r' = A^{\inst I}_z$ as no transferring patterns reach $z$ from $x$. 
    This implies that $A^{\inst I'}_y + R= A^{\inst I}_z$ 
    and the arc-consistency condition holds.

    In the last case $A^{\inst I'}_y$ is a proper subset of $A^{\inst I}_y$,
    and the same for $z$.
    Fix $R\in\constr{R}_{yz}$ and a transferring pattern $r$ from $x$ to $y$.
    As in the previous case we extend $r$ by a step to $z$ using $R$ and call the resulting pattern $r'$.
    If $r'$ is transferring, then $R\cap(A^{\inst I'}_y\times A^{\inst I'}_z)$ is subdirect 
    by the transferring property.
    If $r'$ is not transferring, then $a/\theta_x+r' = A^{\inst I}_z$ 
    for every $a\in A^{\inst I}_x$ which means that $b/\theta_y+R=A^{\inst I}_z$ for every $b\in A^{\inst I}_y$.
    This implies that, for a transferring path $r''$ from $x$ to $z$
    and its extension by $r'''$ using  $-R$ we have $a/\theta_x+r''' = A^{\inst I}_y$ for every $a\in A^{\inst I}_x$. 
    Thus, for every $c\in A^{\inst I}_z$, we have $c/\theta_z-R=A^{\inst I}_y$.
    This gives subdirectness in the last case 
    and proves \Cref{it:binary-subinstance}.

\subsection{Proof of \texorpdfstring{\Cref{prop:linked-gives-absorption-short}}{the proposition}}

This proposition is a version of a standard argument exploiting a linked relation to provide absorption. 
The difference is that, in our case, we have a binary polymorphism which allows us to obtain 2-absorption instead of absorption.
Let us restate the proposition for convenience.

\propLinkedGivesAbsorption*

The remaining part of this section is devoted to the proof of~\Cref{prop:linked-gives-absorption-short}.
Our first goal is to define a \emph{central relation} on $A_x^{\inst I}$ or on $A_y^{\inst I}$: 
a non-trivial binary relation $R\varsubsetneq A\times B$ is \emph{central} if there is $a\in A$  such that $\{a\}\times B\subseteq R$;
the set of all such $a$ is called the \emph{center} of $R$.

Let $p$ be a pattern from $x$ to $y$ such that
$\Conn{p}$ is a linked subrelation of $A_x^{\inst I}\times A_y^{\inst I}$.
Replacing $p$,
if necessary,
by a sufficiently long pattern of the form $k(p-p)$,
we can assume without loss of generality that 
$\Conn{p}$ is not full, 
while $\Conn{p-p}=\Conn{p}-\Conn{p}$
is full~%
(note that if we do need such replacement, 
then the new pattern is 
from $x$ to $x$, not to $y$; 
this is not a problem).

For every  $A\subseteq A_x^{\inst I}$
define $\Cen{A}=\bigcap_{a\in A}(\{a\}+p)$.
Note that $\Cen{\{a\}}-p=A_x^{\inst I}$ for every $a\in A_x^{\inst I}$.
This allows us to fix $A$ which is maximal with respect to the property that $\Cen{A}-p=A_x^{\inst I}$.
If $A=A_x^{\inst I}$, then $(-\Conn{p})$ is a central relation~%
(as any element in $\Cen{A_x^{\inst I}}\neq \emptyset$
can play the role of $a$ in the definition of central relation).
Therefore, we assume that $A$ is a proper subset of $A_x^{\inst I}$.
Define a binary relation $R\subseteq A_x^{\inst I}\times A_x^{\inst I}$
by
\[
(a,b)\in R
\quad\Longleftrightarrow\quad
\text{there exists }c\in \Cen{A}\text{ such that }(a,c),(b,c)\in\Conn{p}.
\]
For every $a\in A$ one has $\{a\}\times A_x^{\inst I}\subseteq R$,
as $a+p\supseteq\Cen{A}$.
Therefore, $A$ is contained in the center of $R$.
By maximality of $A$,
for every $a'\in A_x^{\inst I}\setminus A$ there is some $b'\in A_x$ such that $(a',b')\notin R$.
Thus, $R$ is a proper central relation with center $A$.

Let $R$ be a central relation with center $A$.
Note that $R$ can be a relation between $A_y^{\inst I}$ and $A_x^{\inst I}$
if it was given by $(-\Conn{p})$ or between $A_x^{\inst I}$ and $A_x^{\inst I}$
if it was given by the argument above;
we focus on the first case -- the second case is analogous.
The relation $R$ is compatible with $f_y/f_x$,
since it is defined from the pattern relation by primitive-positive operations.
If
\[
  f_y(A,A_y^{\inst I})\subseteq A,
\]
then,
by commutativity,
$A$ is a proper 2-absorbing subset of $A_y^{\inst I}$.

Otherwise, choose $a\in A$ and $a'\in A_y^{\inst I}\setminus A$ such that
\[
f_y(a,a')=f_y(a',a)\notin A.
\]
By conservativity,
this value must be $a'$.
We claim that
the set $a'+R$
2-absorbs $A_x^{\inst I}$~%
(the case when the RHS of $R$ is $A_x^{\inst I}$ is analogous).
Indeed,
take $b\in a'+R$ and $c\in A_x^{\inst I}$.
Since $a$ is central,
both $(a',b)$ and $(a,c)$ lie in $R$.
Compatibility of $R$ with $f_y/f_x$ gives
\[
(f_y(a',a),f_x(b,c))=(a',f_x(b,c))\in R,
\]
which means that $f_x(b,c)\in a'+R$.
By commutativity the same holds for $f_x(c,b)$.
Therefore, $a'+R$ 2-absorbs $A_x^{\inst I}$ 
and, since $a'$ is not central,
$a'+R$ is a proper subset of $A_x^{\inst I}$; we are done.

\subsection{Proof of \texorpdfstring{\Cref{prop:absorption-gives-subinstance-short}}{the proposition}}

This is standard and known -- in the presence of binary absorption,
arc-consistency is enough to guarantee the existence of a proper subinstance~%
(even though, in our case, the absorption is not by means of the usual
polymorphisms).
We state the proposition for convenience.

\propAbsorptionGivesSubinstance*

\noindent We begin with a simple lemma:
\begin{namedlemma}[Basic 2-absorption facts]\label{lem:absorption-basic}
Let $\inst I$ be arc-consistent and let $\{f_x\}_x$ be a conservative commutative binary polymorphism of $\inst I$.
\begin{itemize}
    \item If $B,C\subseteq A_x$ both 2-absorb $A_x$,
    then $B\cap C\neq\emptyset$.
    \item If $B\subseteq A_x$ 2-absorbs $A_x$ and $p$ is a pattern from $x$ to $y$,
    then $B+p$ 2-absorbs $A_y$.
\end{itemize}
\end{namedlemma}
\begin{proof}
For the first item,
conservativity and commutativity give
\[
f_x(B,C)\subseteq B\cap C,
\]
and the left-hand side is nonempty.
For the second item,
propagate the absorption condition along the constraints of the pattern,
using compatibility of the polymorphism with every constraint.
\end{proof}

We now begin the proof of \Cref{prop:absorption-gives-subinstance-short}:
For every variable $y$,
choose a minimal 2-absorbing subset $A'_y\subseteq A_y$.
Such sets exist because each potato is finite and each full potato 2-absorbs itself.
Moreover, at least for the variable $x$ fixed in the assumption,
the set is proper.
We claim that the sets $A'_y$ form the potatoes of an arc-consistent subinstance.
Let $R\in\constr{R}_{yz}$ be any binary constraint.
Since $\inst I$ is arc-consistent,
$A'_y+R$ is nonempty.
By \namedlemcref{lem:absorption-basic},
it 2-absorbs $A_z$.
Hence,
$
A'_z\cap (A'_y+R) \neq \emptyset
$
and it $2$-absorbs $A_z$ by \namedlemcref{lem:absorption-basic} again.
By the minimality of $A'_z$,
this intersection must be equal to $A'_z$.
Therefore,
\[
A'_z\subseteq A'_y+R.
\]
The same argument for the converse constraint gives the reverse support condition.
Hence, the restricted instance is proper~%
(again, due to the fixed variable $x$) 
and arc-consistent.

\subsection{Proof of \texorpdfstring{\Cref{thm:binary-polymorphism-obstruction}}{the theorem}}

\noindent We restate the theorem for convenience.

\thmBinaryPolymorphismObstruction*

\noindent 
We begin with the forward implication, 
which we prove by contrapositive.
Suppose
there are two dual vertices $(a,b;x)$ and $(b,a;x)$ 
in the same strongly connected component of $G(\inst I)$.
Let $\{f_x\}_x$ be a conservative binary polymorphism.
If
\[
(a,b;x)\longrightarrow(c,d;y)
\]
is an arc in $G(\inst I)$,
then
\[
f_x(a,b)=a
\quad\Longrightarrow\quad
f_y(c,d)=c.
\]
Otherwise, conservativity gives $f_y(c,d)=d$,
and compatibility with $R\in\constr{R}_{xy}$ would force $(a,d)\in R$,
contrary to the definition of the arc.
This implication extends along directed paths.
Since $(a,b;x)$ and $(b,a;x)$ are strongly connected,
we get
\[
f_x(a,b)=a\quad \Longleftrightarrow\quad f_x(b,a)=b.
\]
This is incompatible with commutativity.
Hence, no conservative commutative binary polymorphism exists.

Conversely,
assume that no two dual vertices lie in the same strongly connected component of $G(\inst I)$.
We construct a conservative commutative binary polymorphism.
The duality operation reverses arcs: 
if there is a directed path from $(a,b;x)$ to $(c,d;y)$,
then there is a directed path from $(d,c;y)$ to $(b,a;x)$.
Hence, the duals of the vertices in any strongly connected component again form
a strongly connected  component.
Moreover, the duality reverses the reachability partial order on the strongly
connected components of $G(\inst I)$.

We now assign values to the operations $f_x$ component by component.
Start with all values undefined.
At each step,
choose a sink strongly connected component in the remaining digraph.
For every vertex $(a,b;x)$ in this component,
set
\[
f_x(a,b)=a.
\]
For every dual vertex $(b,a;x)$,
set
\[
f_x(b,a)=a.
\]
This is well-defined because no strongly connected component contains two dual vertices.
Then remove both the chosen component and its dual component,
and continue.

By construction, each element of $\{f_x\}_x$ is conservative and commutative.
It remains to check compatibility.
Note that in order to verify the compatibility condition for a constraint $R\in\constr{R}_{xy}$,
it suffices to check it for pairs of tuples $(a,c),(b,d)\in R$ such that $(a,d)\notin R$,
as the condition is trivially satisfied otherwise.
This means that we only need to check the condition for arcs of $G(\inst I)$.

Suppose
\[
(a,b;x)\longrightarrow(c,d;y)
\]
is an arc and $f_x(a,b)=a$~%
(i.e. we are assigning in the sink component, not its dual).
If $(c,d;y)$ was still unassigned when the component of $(a,b;x)$ was chosen,
then the sink property forces $(c,d;y)$ to lie in the same strongly connected component,
and hence $f_y(c,d)=c$.
If it had already been assigned,
it must've been assigned in an earlier step as a sink,
since the component of $(a,b;x)$ arrows to the component of $(c,d;y)$.
Thus, $f_y(c,d)=c$ by the construction.

If $f_x(a,b)=b$,
i.e. the value was assigned in the dual/source component,
there are two options based on the relation $R\in\constr{R}_{xy}$ which witnesses the arc.
Either $(b,c)\in R$, 
and then $f_x(a,b)=b$ is compatible with both choices for $f_y(c,d)$;
or $(b,c)\notin R$,
and the same $R$ witnesses
\[
(c,d;y)\longrightarrow(a,b;x),
\]
which means~%
(since the component of $(a,b;x)$ is a source component)
that $(c,d;y)$ is in the same strongly connected component and $f_y(c,d)=d$.
In either case, the compatibility condition is satisfied.

%% file: tractability-rounding-and-cost.tex
\section{Tractability: rounding and cost analysis}
\label{sec:tractability-rounding-cost}

This appendix supplies the proofs for the preprocessing, level construction, and
cost analysis from \Cref{sec:tractability}.  We keep all notation and
definitions introduced there, in particular the LP marginals $\lambda_v(a)$ (cf.
the definition of BLP below),
the threshold-conflict and good-level notions, the level procedure defining
$L(v)$, and the truncated instance $\inst I\leq \inst J$.

\subsection{The ambient cycle-consistent instance}
\label{subsec:ambient-cycle-consistent}

\begin{proposition}[$(2,3)$-consistency gives the ambient instance]
\label{prop:23-cycle-consistency-levels}
The $(2,3)$-consistency preprocessing is solution-preserving.  
Moreover, if no potato is emptied, then
$\inst J$ is arc-consistent and cycle consistent.
\end{proposition}

\begin{proof}
The preprocessing is sound: a deleted tuple has no local extension to some third
variable, whereas every global feasible solution would provide such an
extension.  Therefore no value or tuple used by a feasible solution is ever
deleted.  Hence an empty potato certifies infeasibility.

Arc-consistency is part of the enforced local consistency.  For cycle consistency,
let
\[
        p=x=x_0,x_1,\ldots,x_n=x
\]
be a closed pattern, and fix $a\in A_x^{\inst J}$.  By arc-consistency choose
$a_1\in A_{x_1}^{\inst J}$ with $(a,a_1)\in R_{x_0x_1}^{\inst J}$.  Suppose
$a_i$ has been chosen and $(a,a_i)\in R_{x_0x_i}^{\inst J}$.  Applying
$(2,3)$-consistency to $x_0,x_i,x_{i+1}$ gives $a_{i+1}$ such that
\[
        (a_i,a_{i+1})\in R_{x_ix_{i+1}}^{\inst J}
        \quad\text{and}\quad
        (a,a_{i+1})\in R_{x_0x_{i+1}}^{\inst J}.
\]
At the end $x_n=x_0=x$, and the diagonal relation is equality, so $a_n=a$.  Thus
$p$ connects $a$ to itself.
\end{proof}

\subsection{Global thresholds and variable levels}
\label{subsec:global-thresholds-conflicts-levels}

The \emph{basic linear programming} (BLP)
relaxation~\cite{Kumar11:soda,Kun-BLP} of a $\MinCostHom(\relstr A)$ instance $\inst J$ with
variables $V$, constraints $R_{xy}$ for $x,y\in V$, and cost functions
$c_x:A\to\mathbb{Q}_{\geq 0}$ is given in~\Cref{fig:lp}.

\begin{figure}[ht]
    \centering
\begin{alignat}{4}
  &\text{minimize} &\qquad& \sum_{x \in V}\sum_{a \in A} \lambda_x(a) \cdot c_x(a)\notag \\ 
  &\text{subject to} &\qquad& \sum_{a \in A} \lambda_x(a) = 1, \qquad &&\forall x \in V,\notag\\
  & &\qquad& \sum_{(a,b)\in R_{xy}} \lambda_{xy}(a,b) = 1, \qquad &&\forall x,y\in V,\notag\\
  & &\qquad& \sum_{b:\,(a,b)\in R_{xy}} \lambda_{xy}(a,b) = \lambda_x(a), \qquad &&\forall x,y\in V,\,\forall a \in A,\notag \\
  & &\qquad& \sum_{a:\,(a,b)\in R_{xy}} \lambda_{xy}(a,b) = \lambda_y(b), \qquad &&\forall x,y\in V,\,\forall b \in A,\notag \\
  & &\qquad& \lambda_x(a), \, \lambda_{xy}(a,b) \geq 0, \qquad &&\forall x,y \in V,\,\forall a,b \in A.\notag
\end{alignat}
    \caption{BLP relaxation for an instance of $\MinCostHom(\relstr A)$.}
    \label{fig:lp}
\end{figure}

\begin{namedlemma}[One-level conflict bound]
\label{lem:one-level-conflict-bound}
For every fixed $v$, $a\in A_v^{\inst J}$, and $i\geq 1$,
\[
        \Pr[v\mapsto a\text{ conflicts with }t_i]
        \leq 2^{|A|+i}\lambda_v(a).
\]
\end{namedlemma}

\begin{proof}
For a fixed set $B\subseteq A$, the set of bad thresholds is an interval of length at most
$\lambda_v(a)$.  The sampling interval for $t_i$ has length $2^{-i}$, so the
probability is at most $2^i\lambda_v(a)$.  
There are at most $2^{|A|}$ choices of
$B$.
\end{proof}

\subsection{Computing variable levels}
\label{subsec:computing-variable-level}

\begin{namedlemma}[Fresh-threshold invariant]
\label{lem:fresh-threshold-invariant-levels}
At every non-terminal iteration, if $\mathrm{L}$ is not good for $v$, then the
minimal-marginal witness selected in Step~2 conflicts with one of
$t_{\mathrm{P}+1},\ldots,t_{\mathrm{L}}$, where
$\mathrm{P}$ is the value of $\mathrm{L}$ at the previous
iteration (or $0$ at the first iteration).
\end{namedlemma}

\begin{proof}
At a non-terminal iteration, $\mathrm{L}$ is not good for $v$, so there exists
at least one witness with
\[
\lambda_v(a)<t_{\mathrm{L}}
\quad\text{and}\quad
v\mapsto a\text{ conflicts with one of }t_1,\ldots,t_{\mathrm{L}}.
\]
Step~2 selects among these witnesses one with minimal LP marginal.

Initially $\mathrm{P}=0$, so the claim is immediate.  
For later
iterations: 
let $a_{\mathrm{P}}$ be the witness at the previous level;
clearly $\ell(v,a_{\mathrm{P}})=\mathrm{L}-1$. 
By minimality of the previous selection, every
witness from
the previous iteration had LP marginal at least 
$\lambda_v(a_{\mathrm{P}})>2^{-(\mathrm{L}-1)}\geq t_{\mathrm{L}}$, 
so it cannot satisfy
$\lambda_v(a)<t_{\mathrm{L}}$ at the current iteration.

Hence, the currently selected witness cannot conflict with any of the old
thresholds $t_1,\ldots,t_{\mathrm{P}}$; it must conflict with at least one
newly exposed threshold in
$t_{\mathrm{P}+1},\ldots,t_{\mathrm{L}}$.
\end{proof}

\begin{namedlemma}[Termination]
\label{lem:level-procedure-termination-levels}
For every variable $v$, the level procedure terminates after at most
$|A_v^{\inst J}|+1$ iterations.
\end{namedlemma}

\begin{proof}
If $a$ is selected when the current value is $\mathrm{L}$, then
$\lambda_v(a)<t_{\mathrm{L}}\leq 2^{-(\mathrm{L}-1)}$, so
$\ell(v,a)\geq\mathrm{L}$.  
The next $\mathrm{L}$ is $\ell(v,a)+1$;
the threshold in the next iteration is
\[
        t_{\ell(v,a)+1}\leq 2^{-\ell(v,a)}< \lambda_v(a),
\]
so $a$ cannot be used. 
The thresholds in future iterations are even smaller. 
Thus, every non-terminal iteration selects a new value.
\end{proof}

\subsection{Tail bound and cost analysis}
\label{subsec:level-tail-bound}

\begin{namedlemma}[Tail bound]
\label{lem:level-tail-bound}
There is a constant $C=C(\relstr A)$ such that, for every variable $v$ and every
$i\geq 1$,
\[
        \Pr[L(v)=i]\leq C2^{-i}.
\]
\end{namedlemma}

\begin{proof}
Fix $v$.  Let $a_1,\ldots,a_m$ be a possible witness sequence and put
$r_s=\ell(v,a_s)$.  
The case $m=0$~%
(i.e. no conflicts) 
only gives $L(v)=1$
and in this case we satisfy the condition simply by choosing $C$ large enough.
We will bound the probability that this sequence led to the final level $L(v)=r_m+1$.
The update rule gives $r_1<r_2<\cdots<r_m$~%
(we put $r_0=0$ for convenience of notation)
 and the levels tried in the assignment are 
$r_1+1<r_2+1<\cdots<r_m+1$.

For the first witness, by~\namedlemcref{lem:one-level-conflict-bound},
\[
        \Pr[a_1\text{ is chosen}]\leq
        \Pr[v\mapsto a_1\text{ conflicts with }t_1]
        \leq 2^{|A|+1}\lambda_v(a_1)
        \leq 2^{|A|+1}2^{-(r_1-1)}.
\]
For $s\geq 2$, condition on the previous history.  
By~\namedlemcref{lem:fresh-threshold-invariant-levels},
$a_s$ must conflict with a threshold of index between $r_{s-2}+2$ and $r_{s-1}+1$.  
Hence,
\begin{align*}
\Pr[a_s\text{ is chosen}&\mid\text{previous history}]\leq 
    \Pr[v\mapsto a_s\text{ conflicts with }t_{r_{s-2}+2},\ldots,t_{r_{s-1}+1}]\leq \\
    &\leq\ \sum_{j\leq r_{s-1}+1}2^{|A|+j}\lambda_v(a_s)
    \leq\ 2^{|A|+3}2^{r_{s-1}-r_s}.
\end{align*}
As the conflicting conditions involve disjoint intervals of thresholds, 
the events are independent.
By multiplying the probabilities we derive that 
each sequence ending in $r_m = L(v)-1$ has probability at most
\[
        K^m2^{-r_m}
\]
for a constant $K=K(|A|)$.  
Since witnesses are distinct in each sequence, and they come from $A$,
the number of possible witness sequences is bounded by a constant depending only on
$|A|$.  Summing over all sequences gives the claim.
\end{proof}

\begin{namedlemma}[Cost of any solution of the truncated instance]
\label{lem:any-solution-cost}
Assume that $\inst I$ has a solution.  Let $\alpha$ be any solution of $\inst I$.  Then
\[
        \mathbb E\left[\sum_v c_v(\alpha(v))\right]
        \leq C'\sum_v\sum_{a\in A_v^{\inst J}} c_v(a)\lambda_v(a),
\]
where $C'=C'(\relstr A)$.
\end{namedlemma}

\begin{proof}
Fix $v,a$;  
if $a$ survives, then
$\lambda_v(a)\geq t_{L(v)}$,
and further $L(v)\geq \ell(v,a)$.  
By \namedlemcref{lem:level-tail-bound},
\[
        \Pr[a\in A_v^{\inst I}]
        \leq \Pr[L(v)\geq \ell(v,a)]
        \leq \sum_{i\geq \ell(v,a)}C2^{-i}
        \leq 2C2^{-\ell(v,a)}
        \leq 2C\lambda_v(a).
\]
Since $\alpha(v)=a$ implies $a\in A_v^{\inst I}$,
\[
        \Pr[\alpha(v)=a]\leq 2C\lambda_v(a).
\]
Summing over all variables and values against the unary costs gives the result.
\end{proof}

\propRoundedInstanceApproximation*

\begin{proof}[Proof of \Cref{prop:rounded-instance-is-arc-consistent}]
Use \namedlemcref{lem:any-solution-cost} and $\mathrm{BLP}(\inst J)\leq\mathrm{OPT}$.
\end{proof}

%% file: tractability-solvability.tex
\section{Tractability: solvability of the truncated instance}
\label{sec:tractability-solvability}

We prove that the truncated instance $\inst I$ has a solution.  The proof is
existential: we use universal covers as proof objects and compactness.  The
polynomial-time construction of such a solution is omitted here, as it was
discussed at the end of~\Cref{sec:tractability}.
Recall that in this proof we work with a simplified instance.

The proof in this section has parts similar to the reasoning used to establish
bounded path-width dualities in~\cite{DalmauKrokhinMajorityPathwidth}.

\subsection{Two basic propagation facts and cover trees}
\label{subsec:solvability-propagation}

\begin{namedlemma}[Mass survives higher-level truncation]
\label{lem:mass-survives-higher-level-truncation}
Let $z$ be a variable with $L(z)\geq i$, and let
$S\subseteq A_z^{\inst J}$ satisfy $\lambda_z(S)\geq t_i$.  Then
\[
        \lambda_z(S\cap A_z^{\inst I})\geq t_i.
\]
\end{namedlemma}

\begin{proof}
Suppose not, and put $D=S\setminus A_z^{\inst I}$.  Then every $b\in D$ satisfies
$\lambda_z(b)<t_{L(z)}$.  Since $L(z)\geq i$, the threshold $t_i$ is one of
$t_1,\ldots,t_{L(z)}$.  Adding the elements of $D$ one by one to
$S\cap A_z^{\inst I}$ crosses $t_i$, so some $b\in D$ conflicts with $t_i$.  This
contradicts that level $L(z)$ is good for $z$.
\end{proof}

\begin{namedlemma}[LP support propagation]
\label{lem:lp-support-propagation}
Let $R\subseteq A_u^{\inst J}\times A_v^{\inst J}$ be a binary constraint and let
$S\subseteq A_u^{\inst J}$. Then $\lambda_v(S+R)\geq \lambda_u(S)$.
\end{namedlemma}

\begin{proof}
Use the LP distribution on $R$.  The mass of pairs whose first coordinate lies in
$S$ is $\lambda_u(S)$, and all such pairs have second coordinate in $S+R$.
\end{proof}

A \emph{tree pattern} in a binary instance $\inst K$ is a finite tree whose
vertices are labeled by variables of $\inst K$ and whose edges are labeled by
binary constraints of $\inst K$ in a way consistent with the endpoint labels.
A \emph{universal covering pattern for $\inst K$} is an infinite tree pattern that covers all finite tree patterns in $\inst K$.
Every copy of a variable $x$ is incident, for every variable $y$ and every $R\in\constr{R}_{xy}^{\inst K}$, with an edge labelled by $R$ leading to a copy of $y$.
A realization of a tree pattern assigns to every vertex a value from the potato
of its label so that every edge is satisfied by the corresponding binary
constraint.

Sometimes, for convenience, we will treat a tree pattern over $\inst K$ as an instance over the same template as $\inst K$.
Using this viewpoint, each variable of a tree pattern is a copy of a variable of $\inst K$, and each constraint of the tree pattern is a copy of a constraint of $\inst K$.
A realization of the tree pattern is then a solution of the corresponding instance.

\begin{namedlemma}[Tree extension and arc-consistency]
\label{lem:tree-extension-from-ac}
Let $\inst K$ be a binary instance. The following are equivalent:
\begin{enumerate}
        \item there exists $\inst I\leq \inst K$ such that $\inst I$ is arc-consistent and nonempty;
        \item every finite tree pattern in $\inst K$ has a realization;
        \item a universal covering pattern for $\inst K$ has a realization.
\end{enumerate}
Moreover, if $\inst I\leq \inst K$ is arc-consistent and nonempty, then for every
finite tree pattern $t$ in $\inst K$, every vertex $x$ of $t$, and every
$a\in A_{\ell(x)}^{\inst I}$ (where $\ell(x)$ is the variable-label of $x$), there is a
realization of $t$ in $\inst I$ sending $x$ to $a$.
The same holds when $t$ is a universal covering pattern for $\inst K$.
\end{namedlemma}
\begin{proof}
        The last two items are equivalent;
        one implication is trivial, and the other follows by compactness of the tree pattern and the finiteness of the instance $\inst K$.
        The first item implies the second by induction on the number of vertices in the tree pattern,
        using arc-consistency to extend a realization of a smaller tree pattern to a larger one.
        Finally, the third item implies the first by restricting each potato to the set of values that occur in some realization of the universal covering pattern.

        For the moreover part, root $t$ at $x$ and orient all edges away from $x$.
        Assign $x\mapsto a$.  Then process vertices in any order extending this orientation:
        when an edge from a processed parent $u$ to an unprocessed child $v$ is labeled
        by $R$, arc-consistency of $\inst I$ gives
        \[
                (A_u^{\inst I}\times A_v^{\inst I})\cap R
        \]
        with full projection on $A_u^{\inst I}$, so from the already chosen value of $u$
        we choose a value of $v$ in $A_v^{\inst I}$ supported by this edge.  Because $t$
        is a tree, each vertex has a unique parent, so these choices are independent and
        yield a realization of all edges.  Thus we get a realization of $t$ in $\inst I$
        with $x$ mapped to $a$.

        If $t$ is a universal covering pattern, fix a root vertex $x$ and
        a value $a\in A_{\ell(x)}^{\inst I}$.  Every finite rooted subtree
        $s \subseteq t$ containing $x$ has a realization in $\inst I$ sending
        $x$ to $a$ by the finite case above.  By compactness, these compatible
        finite rooted realizations extend to a realization of all $t$ with
        $x$ mapped to $a$.
\end{proof}

\subsection{Level instances}
\label{subsec:level-instances}

Fix a level $i$ and write
\[
        V_i:=\{v:L(v)=i\}.
\]
For $x,y\in V_i$, let $\mathcal P_i(x,y)$ be the set of patterns
\[
        p=x=x_0,x_1,\ldots,x_n=y
\]
in $\inst J$ whose internal variables have level $>i$.
Length-one patterns are allowed.  
Let $R_p^{\inst J}$ be the relation of pairs
connected by $p$ in $\inst J$, 
and let $R_p^{\inst I}$ be the relation of pairs connected by $p$ in $\inst I$.

Define the unrestricted level-$i$ ambient instance $\inst J_i^+$ as follows.  Its
variables are $V_i$, its potatoes are $A_v^{\inst J}$, and its constraint set $\constr{R}^{\inst J_i^+}$ consists of the
relations $R_p^{\inst J}$ for $p\in\mathcal P_i(x,y)$.  Define $\inst I_i^+$ similarly
using potatoes $A_v^{\inst I}$ and relations $R_p^{\inst I}$.
These two instances have identical ``shape'': same variables and identically  
constrained pairs of variables; only the constraint relations differ.

\begin{namedlemma}[Path-induced extension and arc-consistency]
\label{lem:path-induced-ac}
Let
\[
        q=x_0,x_1,\ldots,x_n
\]
be a path in $\inst I$ such that $L(x_j)\geq i$ for every $j>0$ and $L(x_n)=i$, and let
$b\in A_{x_n}^{\inst I}$.  
Then $q$ has a realization~%
(in $\inst J$, but in $\inst I$ except possibly at $x_0$)
ending in $b$ such that every
$x_j$, $j>0$, 
is assigned a value in its truncated potato.  
If also $L(x_0)\geq i$,
then $x_0$ can be chosen in $A_{x_0}^{\inst I}$ as well.

Consequently, for every $p\in\mathcal P_i(x,y)$, the truncated relation
$R_p^{\inst I}$ is subdirect on $A_x^{\inst I}\times A_y^{\inst I}$.
\end{namedlemma}

\begin{proof}
Run the propagation argument backwards from $\{b\}$ in $\inst J$.
Since $b\in A_{x_n}^{\inst I}$
and $L(x_n)=i$, the initial mass is at least $t_i$.  By
\namedlemcref{lem:lp-support-propagation}, LP mass cannot decrease while moving
backwards along the path in $\inst J$.  
At every vertex $x_j$ with $j>0$, we have $L(x_j)\geq i$,
so \namedlemcref{lem:mass-survives-higher-level-truncation} leaves a nonempty
surviving set in $A_{x_j}^{\inst I}$.  This gives the first claim.  If also
$L(x_0)\geq i$, apply the same truncation at $x_0$ to get a truncated value there;
otherwise, a value from $A_{x_0}^{\inst J}$ can be chosen.

For the consequence, fix $p\in\mathcal P_i(x,y)$.  Applying the first claim to
$p$ gives fullness of the second projection of $R_p^{\inst I}$, and applying it to
the reversed path gives fullness of the first projection.  
Hence,
$R_p^{\inst I}$ is subdirect.
\end{proof}

\noindent Note that $(\inst I_i^+,\inst J_i^+)$ is \emph{not} always an ac-leq-cycle pair:
it is quite possible that 
\[
        R_p^{\inst I} \varsubsetneq R_p^{\inst J}\cap (A_x^{\inst I}\times A_y^{\inst I}),
\]
for a pattern $p$ from $x$ to $y$ in $\mathcal P_i(x,y)$, because some internal variable of $p$ may be truncated in $\inst I$.
As these relations are corresponding constraints of $\inst I_i^+$ and $\inst J_i^+$, we do not have $\inst I_i^+\leq \inst J_i^+$ in general.

\begin{namedlemma}[Unrestricted level instances]
\label{lem:level-instance-ac-leq-cycle}
For every level $i$, the instance $\inst I_i^+$ is arc-consistent and
$\inst J_i^+$ is cycle-consistent.
\end{namedlemma}

\begin{proof}
Arc-consistency of $\inst I_i^+$ follows from
\namedlemcref{lem:path-induced-ac}: every constraint $R_p^{\inst I}$ is subdirect
on its endpoint potatoes.

For cycle consistency of $\inst J_i^+$, each constraint of $\inst J_i^+$ is
defined by some pattern in $\inst J$; therefore every closed pattern in
$\inst J_i^+$ expands to a closed pattern in the cycle-consistent instance
$\inst J$, which connects every value to itself.
\end{proof}

\subsection{Induction invariant and a plan for the endgame}
\label{subsec:induction-invariant}

We proceed to solve the instance $\inst I$; 
we solve the levels in increasing order.  
After all levels $<i$ have been solved, we maintain
the following invariant.  If $u,w$ are already solved and
\[
        L(u)\leq L(w)<i,
\]
including the case $u=w$, then for every path
\[
        q=u=x_0,x_1,\ldots,x_n=w
\]
whose internal variables have level at least $i$, the fixed values $s(u)$ and
$s(w)$ are connected by $q$ inside the truncated instance $\inst I$.  
For $u=w$, this says that
$s(u)$ is connected to itself along every such closed path.

The general plan is to solve, 
at level $i$,
the instance $\inst I^+_i$ 
with some additional constraints.
Note that, due to the way constraints of $\inst I^+_i$ are defined,
solving $\inst I^+_i$ already guarantees that:
\begin{itemize}
        \item we obtain a proper solution of $\inst I$ on $V_i$,
        as the paths of length one are included in $\mathcal P_i(x,y)$, and thus they copy the original constraints of $\inst I$;
        \item the invariant will automatically hold if $u$ and $w$ have $L(u)=L(w)=i$; 
        this is because of the constraints defined by the paths in $\mathcal P_i(x,y)$;
        \item the invariant holds for $u$ and $w$ with $L(u)<i$ and $L(w)<i$, by the induction hypothesis.
\end{itemize} 
What is missing is the invariant for $u$ and $w$ with $L(u)<i$ and $L(w)=i$, 
which in particular implies that the solution of $\inst I^+_i$ 
is compatible with the already solved lower levels~%
(again, due to paths of length one). 

In the end the proof will use the tractable-side assumption of \Cref{thm:HardnessMain},
and for this we need an ac-leq-cycle pair.
The $(\inst I^+_i,\inst J^+_i)$ pair is not an ac-leq-cycle pair, 
even though arc-consistency of $\inst I^+_i$ is already established by \namedlemcref{lem:level-instance-ac-leq-cycle}.
The missing condition for an ac-leq-cycle pair is 
not very problematic either;
our problem is that solving $\inst I^+_i$ does not prove the full invariant.
To address the issue, 
we establish arc-consistency of a yet smaller instance.
We will do it by showing solvability of the universal cover of $\inst I^+_i$ with additional constraints. 
Then \namedlemcref{lem:tree-extension-from-ac} will give us arc-consistency of the smaller instance.
We will overcome the ``not an ac-leq-cycle pair'' issue 
and use the tractable-side assumption of \Cref{thm:HardnessMain} to solve the smaller instance.
This will finish the induction step and the proof of \Cref{prop:rounded-instance-has-solution}.

\subsection{Augmented universal covers}
\label{subsec:augmented-covers}

Take the universal cover of $\inst I_i^+$ and call it $\inst T$.
For every node $w'$ of this cover, 
which is a copy of $w\in V_i$,
attach fresh copies of all paths in $\inst I$ defined as
\[
        u=x_0,x_1,\ldots,x_n=w
\]
with $L(u)<i=L(w)$ and all internal vertices at level $>i$,
with endpoint $w$ identified with $w'$.  
Distinct attachments use distinct fresh copies, even if they have the same
label.
Call the augmented tree \emph{$\inst T'$}.

The goal is to find a solution to $\inst T'$ that 
makes each $u$ from above \emph{pinned}, 
i.e. assigned the value $s(u)$ chosen in the previous steps of the procedure.

\begin{namedlemma}[Augmented covers are realizable]
\label{lem:finite-augmented-covers-realizable}
        The instance $\inst T'$ has a pinned solution.
\end{namedlemma}

\begin{proof}
        Call a realization of $\inst T'$ \emph{admissible} if every vertex of
        $\inst T'$ except lower-level leaves is assigned a truncated value (so all
        those constraints are realized inside $\inst I$); lower-level leaves may
        be ambient unless pinned.

        For a finite set $P$ of lower-level leaves, let $\mathsf{S}(P)$ be the
        statement: ``$\inst T'$ has an admissible realization that pins every
        $u\in P$ to $s(u)$.''  We prove $\mathsf{S}(P)$ by induction on $|P|$.

        Base case $P=\varnothing$.  Realize the cover part $\inst T$ in
        $\inst I_i^+$ by \namedlemcref{lem:tree-extension-from-ac}:  
        then for each
        attached path
        \[
                u=x_0,x_1,\ldots,x_n=w,
        \]
        use \namedlemcref{lem:path-induced-ac} from the already fixed value at
        $w$: this realizes the whole path with $x_j$ ($j>0$) truncated and allows
        the leaf $u=x_0$ to be ambient.  Doing this for all attachments yields an
        admissible realization of $\inst T'$, so $\mathsf{S}(\varnothing)$ holds.

        Assume $\mathsf{S}(Q)$ for all finite $Q$ with $|Q|<m$, and fix $P$ with
        $|P|=m$.
        If $m=1$, write $P=\{u\}$ and let
        \[
                p:u=x_0,x_1,\ldots,x_n=w
        \]
        be the attached path of $u$.  By the induction invariant in the closed
        case ($u=u$), $p-p$ connects $s(u)$ to itself in $\inst I$;
        hence $p$ is realizable in $\inst I$ from $u=s(u)$ to some
        $b\in A_w^{\inst I}$.  Keep this realization on $p$, realize the cover
        part of $\inst T$ from root value $b$ at $w$ using the rooted extension
        part of \namedlemcref{lem:tree-extension-from-ac}, and fill all other
        attachments as in the base case.  This gives $\mathsf{S}(P)$.

        If $m=2$, say $P=\{u,u'\}$ with attached paths $p$ (to $w$) and $p'$ (to
        $w'$).  Let $q$ be the unique path in $\inst T$ from $w$ to $w'$, and let
        $q'$ be its unfolding to $\inst I$.  The path
        $p+q'-p'$ has all internal variables at level at least $i$,
        so by the induction invariant it connects $s(u)$ to $s(u')$ in $\inst I$.
        Hence, we can fix compatible values on $p$, $q$~%
        (which is the ``folded'' version of $q'$), 
        and $p'$ with both leaves
        pinned, then extend to all remaining parts exactly as above by rooted
        extensions in $\inst T$ and path-induced fillings on attachments.  So
        $\mathsf{S}(P)$ holds.

        Now let $m\ge 3$.  Choose distinct $u_1,u_2,u_3\in P$.  For each
        $r\in\{1,2,3\}$, by the induction hypothesis for
        $P\setminus\{u_r\}$ obtain an admissible realization $\psi_r$ pinning
        every leaf in $P\setminus\{u_r\}$.  
        For each chosen leaf $u_r$, only $\psi_r$ may be wrong at
        $u_r$, while the other two are correct there; all leaves in
        $P\setminus\{u_1,u_2,u_3\}$ are correct in all three realizations.
        Apply the conservative majority pointwise to $\psi_1,\psi_2,\psi_3$.
        Because majority is a polymorphism, all constraints stay satisfied; by
        conservativity, truncated values at non-leaf vertices are preserved; and every leaf in $P$
        is pinned by the two-out-of-three argument.  So $\mathsf{S}(P)$ holds.

        Thus $\mathsf{S}(P)$ holds for every finite pin set $P$.  Compactness on
        the family of pin equations $u=s(u)$ (over all lower-level leaves $u$)
        now yields one admissible realization of $\inst T'$ pinning all leaves,
        i.e. a pinned solution of $\inst T'$.
\end{proof}

\subsection{Compatible subinstances and completion}
\label{subsec:compatible-subinstances-completion}
The next step is to define an instance $\inst K_i$ such that $\inst K_i\leq \inst I_i^+$,
and so that each solution of $\inst K_i$ would satisfy the induction invariant 
with all already solved lower levels.
For $v\in V_i$, define $A_v^{\inst K_i}\subseteq A_v^{\inst I}$ to be the set of values
$a$ such that there is a pinned solution of the augmented universal cover $\inst T'$ 
sending a copy of $v$ to $a$.
Note that, by \namedlemcref{lem:finite-augmented-covers-realizable}, each $A_v^{\inst K_i}$ is nonempty.
Then put $\constr{R}^{\inst K_i}$ to be the inherited constraints from $\inst I_i^+$, i.e.
\[
        \constr{R}_{xy}^{\inst K_i}:=\{R\cap A_{x}^{\inst K_i}\times A_{y}^{\inst K_i} : R\in\constr{R}_{xy}^{\inst I_i^+}\}.
\]
We proceed to investigate $\inst K_i$:
\begin{namedlemma}[Compatible subinstance is arc-consistent]
\label{lem:compatible-ac-subinstance}
        The instance $\inst K_i$ is arc-consistent.
\end{namedlemma}

\begin{proof}
Fix $R\in\constr{R}_{xy}^{\inst I_i^+}$ and $Q:=R\cap(A_x^{\inst K_i}\times
  A_y^{\inst K_i})$; for any $a\in A_x^{\inst K_i}$, choose a pinned solution
  $\alpha$ of $\inst T'$ sending some copy $x'$ of $x$ to $a$, follow an
  incident cover edge labeled by $R$ to a copy $y'$ of $y$, and put $b:=\alpha(y')$, so $(a,b)\in R$, $b\in A_y^{\inst K_i}$, and hence $(a,b)\in Q$, proving that every $a\in A_x^{\inst K_i}$ has a partner in $A_y^{\inst K_i}$ through $Q$.
The same argument with $x,y$ swapped gives the symmetric statement for every $b\in A_y^{\inst K_i}$, so every constraint of $\inst K_i$ is subdirect on its endpoint potatoes, i.e. $\inst K_i$ is arc-consistent.
\end{proof}

\begin{namedlemma}[Compatible subinstance is solvable]
\label{lem:compatible-subinstance-solvable}
        The instance $\inst K_i$ has a solution.
\end{namedlemma}

\begin{proof}
        In this proof we apply the tractable-side assumption of \Cref{thm:HardnessMain},
        so we need to construct an ac-leq-cycle pair.
        The required ac-leq-cycle pair is obtained by ``unfolding'' constraints of $\inst K_i$ and their ambient counterparts in $\inst J_i^+$.

        For each pair $x,y\in V_i$ and each distinct constraint $R\in\constr{R}_{xy}^{\inst K_i}$,
        choose one path $p\in\mathcal P_i(x,y)$ such that
        $R=R_p^{\inst I}\cap(A_x^{\inst K_i}\times A_y^{\inst K_i})$.
        Only finitely many paths are selected, since there are finitely many distinct binary constraints.
        We unfold their ambient counterparts first:
        we keep the variables $V_i$ and their potatoes,
        and for every selected path $p$, we introduce a fresh copy of its internal variables~%
        (if $p$ has no internal variables we simply introduce a relation into the corresponding
        $\constr{R}_{xy}$)
        and place along every gadget edge the corresponding original binary constraint from $\inst J$; all potatoes are the ambient ones.
        Call the resulting instance $\widetilde{\inst J}_i$.

        Next we unfold $\inst K_i$;
        the goal is to obtain an instance ``similar'' to $\widetilde{\inst J}_i$ 
        i.e. the same variable set, the same pairs of variables constrained etc.
        To do so, we use exactly the same selected paths and copies of their internal variables:
        we keep the variables $V_i$
        and for every selected path $p:x=x_0,x_1,\dotsc,x_n=y$ we use the corresponding gadget,
        but for each $x_i$ we restrict the potato $A_{x_i}^{\inst I}$ to those
        values that participate in a realization of $p$ in $\inst I$ whose endpoints lie in $A_x^{\inst K_i}$ and $A_y^{\inst K_i}$.
        We then place along every gadget edge the corresponding original binary constraint from $\inst I$, but restricted to the potato of the corresponding gadget variable.
        Call the resulting instance $\widetilde{\inst K}_i$.

        The instance $\widetilde{\inst K}_i$ is a subinstance of $\widetilde{\inst J}_i$,
        since every gadget edge-constraint of $\inst I$ is inherited from $\inst J$,
        and we then only shrink endpoint potatoes further.
        Moreover, $\widetilde{\inst K}_i$ is arc-consistent, since every constraint is subdirect on its endpoint potatoes by \namedlemcref{lem:compatible-ac-subinstance} 
        in the case of variables in $V_i$,
        and for every gadget edge $(x_{j-1},x_j)$ internal endpoint values are
        chosen exactly from realizations of the unfolded path, so each allowed
        value on one endpoint has a supporting value on the other endpoint.
        Finally, $\widetilde{\inst J}_i$ is cycle-consistent, since it is an unfolding of constraints of $\inst J$, which is cycle-consistent.
        Thus $(\widetilde{\inst K}_i,\widetilde{\inst J}_i)$ is an ac-leq-cycle pair.
        The relations produced by $(2,3)$-consistency and their restrictions to potatoes are pp-definable over $\relstr{A}$,
        hence belong to $\relstr{A}$ by our convention in~\Cref{sec:result}; thus both unfolded instances are instances of $\CSP(\relstr{A})$.
        By the tractable-side assumption of \Cref{thm:HardnessMain},
        $\widetilde{\inst K}_i$ has a solution.
        By construction, this gives a solution of $\inst K_i$.
\end{proof}

To complete the induction step, 
we argue that every solution of $\inst K_i$ 
satisfies the missing part of the induction invariant.
Take variables $u$ and $w$ with $L(u)<i=L(w)$, and let $p:u=x_0,x_1,\ldots,x_n=w$ be a path in $\inst I$ whose internal variables have level $>i$.
Take a solution of $\inst K_i$ and let $s(w)$ be the value assigned to $w$.
By the definition of $\inst K_i$, there is a pinned solution of the augmented universal cover $\inst T'$ sending a copy of $w$ to $s(w)$.
The pinned solution of $\inst T'$ realizes a copy of $p$ in $\inst I$ whose endpoint $u$ is assigned to $s(u)$, so the induction invariant holds for $u$ and $w$.

This finishes the argument: after all the steps we have a solution to $\inst I$ satisfying the constraints in levels~(because we solved augmented instances in the levels)
and across the levels~(because we maintained the induction invariant).
\Cref{prop:rounded-instance-has-solution} is proved.

%% file: cycles-are-hard.tex
\section{Inapproximability of \texorpdfstring{$\cUGC(k)$}{cUGC(k)}}
\label{sec:directed-cycle-gap}

We show that, for any $k\geq 2$, $\cUGC(k)$ is NP-hard to approximate within any constant factor
assuming the UGC. For $k=2$ this follows from~\cite{Khot07:sicomp}.
Our strategy is to use Raghavendra's celebrated result, which turns SDP
integrality gaps into UGC hardness~\cite{Raghavendra08:everycsp}.

In this section, let $[k]=\{0,\dots,k-1\}$. Recall from~\Cref{sec:intro} that $Z_k=\{(a,a+1 \bmod k):a\in [k]\}$; so
$([k],Z_k)$ is the directed $k$-cycle.
An instance of $\cUGC(k)$ is a directed graph $G=(V,E)$. The cost incurred by an
assignment $\alpha:V\to [k]$ on an edge $(u,v)$ 
is
\[
  \operatorname{cost}_{(u,v)}(\alpha):=
    \begin{cases}
    0, & \text{if } (\alpha(u),\alpha(v))\in Z_k,\\
    1, & \text{otherwise.}
    \end{cases}
\]
The objective is to minimize the normalized cost
\[
  \operatorname{cost}_G(\alpha):=
    \frac1{|E|}\sum_{(u,v)\in E}\operatorname{cost}_{(u,v)}(\alpha).
\]
We write
\[
  \OPT_{\min}(G):=\min_{\alpha:V\to [k]}\operatorname{cost}_G(\alpha),
    \qquad
    \OPT_{\max}(G):=1-\OPT_{\min}(G).
\]

\subsection{SDP relaxation}

The \emph{basic semidefinite programming} (SDP)
relaxation~\cite{Raghavendra08:everycsp} for maximizing the number of satisfied
constraints in a $\cUGC(k)$ instance is given in~\Cref{fig:sdp}. It has,
for each vertex $u\in V$ and each label $a\in [k]$, a
vector $\blambda_u(a)\in\mathbb R^q$, where the dimension $q$ is arbitrary and 
not fixed in advance, and also a unit vector $\blambda_0\in\mathbb
R^q$.

\begin{figure}[ht]
\centering
  \begin{alignat}{4}
    \SDP_{\max}(G) :=\ &\ \text{maximize} &\qquad& \frac1{|E|}\sum_{(u,v)\in E}\sum_{a=0}^{k-1} \langle \blambda_u(a),\blambda_v(a+1\bmod k)\rangle \notag \\
&\ \text{subject to} & \qquad& \langle \blambda_u(a),\blambda_u(b)\rangle=0, \qquad && \forall u\in V,\,\forall a\neq b\in [k], \notag\\
    & & \qquad& \langle \blambda_u(a),\blambda_v(b) \rangle \geq 0, \qquad && \forall
    (u,v)\in E,\,\forall a,b\in [k],\notag\\
    & & \qquad & \sum_{a=0}^{k-1}\blambda_u(a)=\blambda_0, \qquad && \forall u\in V,\notag \\
    & &  \qquad& \|\blambda_0||^2 = 1.\notag
\end{alignat}
  \caption{SDP relaxation for an instance $(V,E)$ of $\cUGC(k)$.}
  \label{fig:sdp}
\end{figure}
We let $\SDP_{\min}(G):=1-\SDP_{\max}(G)$.

\subsection{Gap instances}

Let $n\equiv 1\pmod k$, and let $G_n$ be the directed $n$-cycle with vertices
$[n]$ and arcs $(u,u+1\bmod n)$, $u\in [n]$. As an instance of $\cUGC(k)$, we
have
\[
    \OPT_{\min}(G_n)=\frac1n,
    \qquad
    \OPT_{\max}(G_n)=1-\frac1n.
\]
On the one hand we cannot satisfy all the constraints by the choice of $n$, 
and on the
other hand, deleting any one edge leaves a directed path and the labels can then be propagated along the path.

We will show below that $\SDP_{\min}(G_n)=O_k(n^{-2})$. It will be convenient to
work with complex (rather than real) vectors. However, this is only a
notational change as a complex vector $z\in\mathbb C^q$ can
be replaced by its realification $\mathcal
R(z)=(\operatorname{Re}z,\operatorname{Im}z)\in\mathbb R^{2q}$.
One can verify that $\langle \mathcal R(z),\mathcal R(w)\rangle =\operatorname{Re}\langle z,w\rangle_{\mathbb C}$ and $\|\mathcal R(z)\|_2=\|z\|_2$.
In the construction below, all relevant inner products of the
SDP vectors are real and nonnegative, so realification gives a valid real solution.

\begin{namedlemma}[SDP solution for an inconsistent cycle]
\label{lem:sdp-cycle-value-direct}
For the instance $G_n$ defined above, $\SDP_{\min}(G_n)= O_k(n^{-2})$.
\end{namedlemma}

\begin{proof}
For $r\in [k]$, set $b_r=(2\pi r)/k$.
For a vertex $u\in [n]$ and a label $a\in [k]$, define $\bkappa_u(a)\in\mathbb{C}^k$ by
\[
\bkappa_u(a)(r)=\frac{1}{\sqrt k}\exp\!\left(i b_r\left(a-u+\frac{u}{n}\right)\right),
  \qquad r\in [k].
\]
For each fixed $u$, the vectors $\bkappa_u(0),\dots,\bkappa_u(k-1)$ are an orthonormal basis, since
\[
\langle \bkappa_u(a),\bkappa_u(b)\rangle
  =\frac{1}{k}\sum_{r\in [k]}\exp\!\left(\frac{2\pi i r(b-a)}{k}\right)
=
\begin{cases}
1, & a=b,\\
0, & a\ne b.
\end{cases}
\]
Define the SDP vector $\blambda_u(a)\in\mathbb{C}^{k^2}$ by coordinates
  indexed by pairs $(x,y)\in [k]^2$ as follows:
\[
  \blambda_u(a)(x,y)=\frac{1}{\sqrt k}\,\bkappa_u(a)(x)\overline{\bkappa_u(a)(y)}.
\]
Define $\blambda_0\in\mathbb{C}^{k^2}$ by
\[
 \blambda_0(x,y)=
 \begin{cases}
 1/\sqrt k, & x=y,\\
 0, & x\neq y.
 \end{cases}
\]
We now check the SDP constraints.

First, $\|\blambda_0\|=1$. Also, for every $u$ and every coordinate $(x,y)$,
\[
\sum_{a\in [k]} \blambda_u(a)(x,y)
=\frac{1}{\sqrt k}\sum_{a\in [k]}\bkappa_u(a)(x)\overline{\bkappa_u(a)(y)}
=\frac{1}{k\sqrt k}\exp\!\left(-i(b_x-b_y)\left(u-\frac{u}{n}\right)\right)
\sum_{a\in [k]}\exp\!\bigl(i(b_x-b_y)a\bigr).
\]
The final sum is $k$ if $x=y$ and $0$ otherwise. Hence
\[
  \sum_{a\in [k]} \blambda_u(a)=\blambda_0.
\]
Second,
\[
  \langle
  \blambda_u(a),\blambda_u(b)\rangle=\frac{1}{k}\,\bigl|\langle\bkappa_u(a),\bkappa_u(b)\rangle\bigr|^2.
\]
Therefore
\[
  \langle \blambda_u(a),\blambda_u(b)\rangle=0 \qquad (a\ne b),
\]
  and $\|\blambda_u(a)\|^2=1/k$.

Third, for any vertices $u,v$ and labels $a,b$,
\[
  \langle
  \blambda_u(a),\blambda_v(b)\rangle=\frac{1}{k}\,\bigl|\langle\bkappa_u(a),\bkappa_v(b)\rangle\bigr|^2\ge 0.
\]

It remains to bound the objective. For an edge $(u,u+1\bmod n)\in E$ we
compare $\blambda_u(a)$ and $\blambda_v(a+1 \bmod k)$ as these are the only terms in the
objective. First,
\[
\begin{aligned}
  \langle\bkappa_u(a),\bkappa_{u+1\bmod n}(a+1 \bmod k)\rangle
  &=\frac{1}{k}\sum_{r\in [k]}
\exp\!\left(i b_r\left[a+1-(u+1)+\frac{u+1}{n}-a+u-\frac{u}{n}\right]\right)\\
  &=\frac{1}{k}\sum_{r\in [k]}\exp\!\left(i\frac{b_r}{n}\right) =
\frac{1}{k}\sum_{r\in [k]}\exp\!\left(\frac{2\pi i r}{kn}\right) =
  \frac{1}{k}\sum_{r\in [k]} e^{ib_r/n}.
\end{aligned}
\]
Denote the last sum by $\ell$. Then
\[
  \langle \blambda_u(a),\blambda_{u+1\bmod n}(a+1\bmod
  k)\rangle=\frac{1}{k}|\ell|^2.
\]
Summing over all $a\in [k]$, every edge contributes $|\ell|^2$. Hence
$\SDP_{\max}(G_n)\ge |\ell|^2$.
Using Taylor expansion for fixed $k$,
\[
 e^{i b_r/n}=1+\frac{i b_r}{n}-\frac{b_r^2}{2n^2}+O_k(n^{-3}).
\]
Thus, with
\[
  M_1=\frac{1}{k}\sum_{r\in [k]} b_r,
\qquad
M_2=\frac{1}{k}\sum_{r\in [k]} b_r^2,
\]
we get
\[
 \ell=1+\frac{iM_1}{n}-\frac{M_2}{2n^2}+O_k(n^{-3}),
\]
and hence
\[
|\ell|^2=1-\frac{M_2-M_1^2}{n^2}+O_k(n^{-3}).
\]
Now
\[
M_1=\frac{\pi(k-1)}{k},
\]
and
\[
  M_2=\frac{1}{k}\sum_{r\in [k]}\left(\frac{2\pi r}{k}\right)^2
=\frac{2\pi^2(k-1)(2k-1)}{3k^2}.
\]
Therefore
\[
M_2-M_1^2=\frac{\pi^2(k^2-1)}{3k^2}.
\]
Consequently
\[
\SDP_{\min}(G_n)\le \frac{\pi^2(k^2-1)}{3k^2}\cdot\frac{1}{n^2}+O_k(n^{-3}),
\]
and in particular $\SDP_{\min}(G_n)=O_k(n^{-2})$.
\end{proof}

\subsection{From the SDP gap to hardness}

\begin{theorem}
\label{thm:no-constant-factor-directed-cycle}
  Assume the UGC. Then, for every fixed $k\ge 2$, $\cUGC(k)$ admits no
  polynomial-time constant-factor approximation algorithm. \end{theorem}
\begin{proof}
Fix $k\ge 2$.  By~\namedlemcref{lem:sdp-cycle-value-direct}, for sufficiently large
$n\equiv 1\pmod k$, we have
$\SDP_{\min}(G_n)\le \frac{B_k}{n^2}$ for a fixed constant $B_k$.
Equivalently, $\SDP_{\max}(G_n) \ge 1-\frac{B_k}{n^2}$.
Recall that $\OPT_{\max}(G_n)=1-\frac1n$.
We set $c_n=1-\frac{B_k}{n^2}$ and $s_n=1-\frac1n$.

By Raghavendra's result~\cite{Raghavendra08:everycsp}, for every $\varepsilon>0$ it is NP-hard to distinguish instances satisfying
$\OPT_{\max}\ge c_n-\varepsilon$ from instances satisfying $\OPT_{\max}\le s_n+\varepsilon$.
Equivalently,
it is NP-hard to distinguish instances satisfying
$\OPT_{\min}\le 1-c_n+\varepsilon =\frac{B_k}{n^2}+\varepsilon$
from
$\OPT_{\min}\ge 1-s_n-\varepsilon =\frac1n-\varepsilon$.

Take $\varepsilon=\frac{B_k}{n^2}$. Then, it is NP-hard to distinguish $\OPT_{\min}\le \frac{2B_k}{n^2}$
from $\OPT_{\min}\ge \frac1n-\frac{B_k}{n^2}$. The ratio between the two sides is $(\frac1n-\frac{B_k}{n^2})/(\frac{2B_k}{n^2}) = \frac{n-B_k}{2B_k}$, which tends to infinity with $n$.
Consequently, it is NP-hard to approximate $\cUGC(k)$ within any constant factor.
\end{proof}